\documentclass[letterpaper,11pt]{article}
\usepackage[english]{babel}

\usepackage[utf8]{inputenc}
\usepackage{multirow}

\usepackage{amsmath,amssymb,amsthm,stmaryrd}
\usepackage{graphicx,subfig}
\usepackage{algorithm,algorithmic}
\usepackage{color} 
\usepackage{rotating}

\usepackage{fullpage}

\newcommand{\figurePath}[1]{#1.pdf}

\newtheoremstyle{example}{\topsep}{\topsep}%
     {}
     {}
     {\bfseries}
     {.}
     {\newline}
     {\thmname{#1}\thmnumber{ #2}\thmnote{ #3}}

\newtheorem{defn}{Definition}
\newtheorem{thm}{Theorem}
\newtheorem{property}[thm]{Property}
\newtheorem{lem}[thm]{Lemma}
\newtheorem{corollary}[thm]{Corollary}

   \theoremstyle{example}
   \newtheorem{example}{Example}

\newcommand{\DT}[1]{
\protect\raisebox{0.5em}{\scriptsize $\protect\underrightarrow{\;(#1)\;}$}
}

\newcommand{\pre}[0]{\prec}
\newcommand{\preCons}[0]{\triangleleft}

\newcommand{\succI}[0]{succ_I}
\newcommand{\succIp}[0]{succ_{I'}}
\newcommand{\succIph}[0]{succ_{I_{\phi}}}
\renewcommand{\L}[0]{L}
\newcommand{\Lp}[0]{L'}
\newcommand{\B}[1]{\mathcal B_{#1}}
\newcommand{\Bdot}[1]{\mathcal B^{\w}_{#1}}
\newcommand{\Bp}[1]{\mathcal B'_{#1}}

\newcommand{\BB}[0]{\mathcal B}
\newcommand{\sL}[0]{\lvert L\rvert}

\newcommand{\inst}{\ensuremath{\left<\Sigma,T,\psi\right>}}
\newcommand{\instp}{\ensuremath{\left<\Sigma',T',\psi'\right>}}
\renewcommand{\ni}[0]{ni}
\newcommand{\no}[0]{no}

\newcommand{\block}[0]{block_{I,\BB}}
\newcommand{\blockI}[0]{block_{I,\BB}}
\newcommand{\blockIp}[0]{block_{I',\BB'}}

\newcommand{\Condition}[1]{{\it{\em (#1)}}}
\newcommand{\Cii}[0]{\Condition{i}}
\newcommand{\Ciii}[0]{\Condition{ii}}
\newcommand{\Civ}[0]{\Condition{iii}}
\newcommand{\Cv}[0]{\Condition{iv}}

\newcommand{\zz}{}
\newcommand{\emptyinst}{\varepsilon}

\newcommand{\intvl}[2]{\ensuremath{\left\llbracket #1\,;\,#2 \right\rrbracket}}

\newcommand{\w}[0]{\protect\raisebox{0.1em}{$\centerdot$}}
\definecolor{rouge}{rgb}{1.0, 0.0, 0.0}
\definecolor{darkred}{rgb}{0.8, 0.0, 0.0}
\newcommand{\blockFrame}[1]{ \text{\fcolorbox{red}{white}{$ #1$}}}
\newcommand{\blockFont}[1]{\mathsf{#1}}
\title{Sorting by Transpositions is Difficult}
\date{}
\author{Laurent Bulteau, Guillaume Fertin, Irena Rusu\smallskip\\
\small
Laboratoire d'Informatique de Nantes-Atlantique (LINA), UMR CNRS
  6241\\ \small
  Universit\'e de Nantes, 2 rue de la Houssini\`ere, 44322 Nantes 
  Cedex 3 - France \\ \small
\texttt{ \{Laurent.Bulteau, Guillaume.Fertin, Irena.Rusu\}@univ-nantes.fr}
}

\begin{document}
\maketitle
\begin{center}\begin{minipage}{0.75\textwidth}
{\bf Abstract.}  
In comparative genomics, a transposition is an operation that exchanges two 
consecutive sequences of genes in a genome. 
The transposition
distance, that is, the minimum number of transpositions needed to transform a 
genome into another, is, according to numerous studies, 
a relevant evolutionary distance. 
The problem of computing this 
distance when genomes are represented by permutations, called the 
\textsc{Sorting by Transpositions} problem, has been introduced by 
Bafna and Pevzner \cite{BafnaPevzner95} in 1995. 
It has naturally been the focus of a number of studies, but the computational
complexity of this problem has remained undetermined for 15 years.

In this paper, we answer this long-standing open question by proving that the 
\textsc{Sorting by Transpositions} problem is \textsf{NP}-hard. 
As a corollary of our result, we also prove that the following problem \cite{Christie98} is NP-hard: given a permutation $\pi$, is it possible
to sort $\pi$ using $d_b(\pi)/3$ permutations, where $d_b(\pi)$ is the number of breakpoints of $\pi$?

\end{minipage}
\end{center}
\section*{Introduction}

Along with reversals, transpositions are one of the most elementary 
large-scale operations that can affect a genome. A transposition 
consists in swapping two consecutive sequences of genes or, 
equivalently, in moving a sequence of 
genes from one place to another in the genome. 
The transposition distance between two genomes is the minimum number of such 
operations that are needed to transform one genome into the other. 
Computing this distance is a challenge in comparative genomics, 
since it gives a maximum
parsimony evolution scenario between the two genomes.

The \textsc{Sorting by Transpositions} problem is the problem of computing the
transposition distance between genomes represented by permutations.
Since its introduction by Bafna and 
Pevzner~\cite{BafnaPevzner95,BafnaPevzner98}, 
the complexity class of this problem has never been established. Hence
a number of studies 
\cite{BafnaPevzner98,Christie98,GuPC99,
HartmanS06,EliasH06,
Benoit-GagneH07,
FengZ07} 
aim at 
designing approximation algorithms or heuristics, the best known fixed-ratio 
algorithm being a 1.375-approximation~\cite{EliasH06}. 
Other works \cite{
GuyerHV97,Christie98,ErikssonEKSW01,Labarre06,EliasH06,Benoit-GagneH07}
aim at computing bounds on the transposition distance of 
a permutation.
Studies have also been devoted to variants of this problem, by considering, 
for example, prefix transpositions \cite{DiasM02,Labarre08,ChitturiS08}
(in which one of the blocks is a prefix of the sequence),
or distance between strings \cite{ChristieI01,CormodeM02,ShapiraS02,RadcliffeSW05,KolmanW06} 
(where multiple occurences of each 
element are allowed in the sequences), possibly
with weighted or prefix transpositions 
\cite{Qi2006,Bongartz06,AmirABLLPSV06,AmirAILP07,ChitturiS08}.

In this paper, we address the long-standing issue of determining
the complexity class of the \textsc{Sorting by Transpositions} problem, 
by giving a polynomial time
reduction from SAT, thus proving the \textsf{NP}-hardness of this problem.
Our reduction is based on the study of transpositions removing three 
breakpoints. A corollary of our result is the \textsf{NP}-hardness of the
 following problem, introduced by \cite{Christie98}: 
given a permutation $\pi$, is it possible
to sort $\pi$ using $d_b(\pi)/3$ permutations, 
where $d_b(\pi)$ is the number of breakpoints of $\pi$?


\section{Preliminaries}


\subsection{Transpositions and Breakpoints}

In this paper, $n$ denotes a positive integer.
Let $\intvl{a}{b}={\{x\in\mathbb N\mid a\leq x\leq b\}}$,
and $Id_n$ be the identity permutation over \intvl{0}{n}.
We consider only permutations of 
\intvl{0}{n} such that $0$ and $n$ are fixed-points.
Given a word $u_1\ u_2\ \ldots\ u_l$, 
a \emph{subword} is a subsequence 
$u_{p_1}\ u_{p_2}\ \ldots\ u_{p_{l'}}$, where $1\leq p_1<p_2<\ldots <p_{l'}\leq l$. 
A \emph{factor} is a subsequence of contiguous elements, i.e. a subword with $p_{k+1}=p_k+1$ 
for every $k\in\intvl1{l'-1}$.

A transposition is an operation that exchanges two consecutive 
factors of a sequence.
As we only work with permutations, it is defined as a permutation $\tau_{i,j,k}$, which, once composed to a permutation $\pi$, realise this operation 
(see Figure~\ref{fig:ex1a}). 
The transposition $\tau_{i,j,k}$ is formally defined as follows.
\begin{defn}[Transposition] \label{def:tau}
Given three integers $i,j,k$ such that $0< i<j<k \leq n$, the \emph{transposition} 
$\tau_{i,j,k}$ over $\intvl{0}{n}$ is the following permutation (we write $q(j)=k+i-j$):

\begin{align*}
&\text{For any }0\leq x< i,&& \tau_{i,j,k}(x)=x\\
&\text{For any }i\leq x <q(j),&& \tau_{i,j,k}(x)=x+j-i\\
&\text{For any }q(j)\leq x <k,&& \tau_{i,j,k}(x)=x+j-k\\
&\text{For any }k\leq x \leq n,&& \tau_{i,j,k}(x)=x
\end{align*}
\end{defn}

Note that the inverse function of $\tau_{i,j,k}$ is also a transposition. More precisely, $\tau_{i,j,k}^{-1}=\tau_{i,q(j),k}$.

\begin{figure}[t]\begin{center}
$$\begin{array}{ccc}
\pi&=&(\pi_0 \pi_1 \ldots  \pi_{i-1}  \underline{\pi_i \ldots  \pi_{j-1}}  \ \underline{\pi_j \quad \ldots \quad \pi_{k-1}} \pi_k \ldots  \pi_n)\\
\pi\circ\tau_{i,j,k}&=&(\pi_0 \pi_1 \ldots  \pi_{i-1} \underline{\pi_j \quad \ldots \quad \pi_{k-1}} \   \underline{\pi_i \ldots \pi_{j-1}} \pi_k \ldots\pi_n) 
\end{array}$$
\caption{\label{fig:ex1a}Representation of a transposition $\tau_{i,j,k}$, with $0< i<j<k \leq n$.}
\end{center}
\end{figure}

The following two properties directly follow from the definition of 
a transposition:
\begin{property}\label{prop:echangeTau}
Let $\tau=\tau_{i,j,k}$ be a transposition, $q(j)=k+i-j$, and $u,v\in\intvl0n$ be two integers such that $u<v$. 
Then:
\begin{eqnarray*}
\tau(u)>\tau(v)& \Leftrightarrow& i\leq u<q(j) \leq v<k\\
\tau^{-1}(u)>\tau^{-1}(v) &\Leftrightarrow& i\leq u<j \leq v<k
\end{eqnarray*}
\end{property}


\begin{property}\label{prop:tauXmoinsUn}
Let $\tau$ be the transposition $\tau=\tau_{i,j,k}$, and 
write $q(j)=k+i-j$.
For all $x\in\intvl1n$, the values of $\tau(x-1)$ and $\tau^{-1}(x-1)$ are the following:
$$\begin{array}{lrcl}
\forall x \notin\{i,q(j),k\},&\tau(x-1)&=&\tau(x)-1\\
\forall x \notin\{i,j,k\},&\tau^{-1}(x-1)&=&\tau^{-1}(x)-1
\end{array}$$
$$\begin{array}{rclrcl}
\tau(i-1)&=&\tau(q(j))-1\hspace{0.5cm}&\tau^{-1}(i-1)&=&\tau^{-1}(j)-1\\
\tau(q(j)-1)&=&\tau(k)-1&\tau^{-1}(j-1)&=&\tau^{-1}(k)-1\\
\tau(k-1)&=&\tau(i)-1&\tau^{-1}(k-1)&=&\tau^{-1}(i)-1\\
\end{array}$$
\end{property}

\begin{defn} [Breakpoints]
Let $\pi$ be a permutation of $\intvl{0}{n}$. If $x\in\intvl1n$ is an integer such that 
$\pi(x-1)=\pi(x)-1$, then $(x-1,x)$ is \emph{an adjacency} of $\pi$, otherwise it is \emph{a breakpoint}. 
We write $d_b(\pi)$ the number of breakpoints of $\pi$.
\end{defn}

The following property yields that the number of breakpoints of a permutation can be reduced by at most 3 when a transposition is applied:

\begin{property}\label{prop:td-bd}
Let $\pi$ be a permutation and $\tau=\tau_{i,j,k}$ be a transposition 
(with $0<i<j<k\leq n$). 
Then, for all $x\in\intvl{1}{n}-\{i,j,k\}$,
\begin{equation*}
(x-1,x) \text{ is an adjacency of } \pi \Leftrightarrow 
(\tau^{-1}(x)-1,\tau^{-1}(x)) \text{ is an adjacency of } \pi\circ\tau.
\end{equation*}
Overall, we have $d_b(\pi\circ\tau)\geq d_b(\pi)-3$.
\end{property}
\begin{proof}
For all $x\in\intvl{1}{n}-\{i,j,k\}$, we have:
\begin{align*}
(x-1,x) \text{ adjacency of } \pi 
&\Leftrightarrow \pi(x-1)=\pi(x)-1\\
&\Leftrightarrow \pi(\tau(\tau^{-1}(x-1)))=\pi(\tau(\tau^{-1}(x)))-1\\
&\Leftrightarrow \pi\circ\tau(\tau^{-1}(x)-1)=\pi\circ\tau(\tau^{-1}(x))-1\text{ by Prop.~\ref{prop:tauXmoinsUn}}\\
&\Leftrightarrow (\tau^{-1}(x)-1,\tau^{-1}(x)) \text{ adjacency of } \pi\circ\tau.
\end{align*}
\end{proof}


\subsection{Transposition distance}

The transposition distance of a permutation is the minimum number of 
transpositions needed to transform it into the identity. 
A formal definition is the following:
\begin{defn}[Transposition distance]
Let $\pi$ be a permutation of $\intvl0n$. The \emph{transposition 
distance} $d_t(\pi)$ from 
$\pi$ to $Id_n$ is the minimum value $k$ for which there exist $k$ 
transpositions $\tau_1,\tau_2,\ldots,\tau_k$, satisfying:
\begin{equation*}
\pi\circ\tau_k\circ\ldots\circ \tau_2\circ\tau_1=Id_n
\end{equation*}
\end{defn}
 The decision problem of 
computing the transposition distance is the following:
\begin{center}
\framebox{
\begin{tabular}{l}
\noindent {\sc Sorting by Transpositions Problem~\cite{BafnaPevzner95}}
\\ \noindent {\sc Input:} A permutation $\pi$, an integer $k$.
\\ \noindent {\sc Question:}  Is $d_t(\pi)\leq k$?
\end{tabular}
}
\end{center}

The following property directly follows from Property~\ref{prop:td-bd}, 
since for any $n$ the number of breakpoints of $Id_n$ is $0$.
\begin{property}\label{prop:td-bd2}
Let $\pi$ be a permutation, then $d_t(\pi)\geq d_b(\pi)/3$.
\end{property}
Figure~\ref{fig:ex1b} gives an example of the computation of the 
transposition distance.

\begin{figure}[t]\begin{center}$$
\begin{array}{lcl}
\pi&=&0\ \underline{2\ 4}\ \underline{3\ 1}\ 5 \\
\pi\circ\tau_{1,3,5}&=& 0\ \underline{3}\ \underline{1\ 2}\ 4\ 5 \\
\pi\circ\tau_{1,3,5}\circ\tau_{1,2,4}&=&  0\ 1\ 2\ 3\ 4\ 5 \\
\end{array}$$
\caption{\label{fig:ex1b}The transposition distance from 
$\pi=(0\ 2\ 4\ 3\ 1\ 5)$ to $Id_5$ is $2$: 
it is at most $2$ since $\pi\circ\tau_{1,3,5}\circ \tau_{1,2,4}=Id_5$, 
and it cannot be less than $2$ since Property~\ref{prop:td-bd2} applies 
with $d_b(\pi)/3=5/3>1$.}
\end{center}
\end{figure}


\section{3-Deletion and Transposition Operations}

In this section, we introduce 3DT-instances, which are the cornerstone of
our reduction from SAT to the \textsc{Sorting by Transpositions} problem, since 
they are used as an intermediate between instances of the two problems. 
We first define 3DT-instances and the possible operations that can be 
applied to them, then we focus on the equivalence 
between these instances and permutations.


\subsection{3DT-instances}


\begin{defn}[3DT-instance]
\label{def:3DTinstance}
A \emph{3DT-instance} $I=\left<\Sigma,T,\psi\right>$ of span $n$ 
is composed of the following elements:
\begin{itemize}
\item $\Sigma$: an alphabet;
\item $T=\{(a_i,b_i,c_i)\mid 1\leq i\leq \left\lvert T\right\rvert\}$: 
	a set of (ordered) triples of elements of $\Sigma$, partitioning $\Sigma$ 
(i.e. all elements are pairwise distinct, and 
$\bigcup_{i=1}^{|T|} \{a_i,b_i,c_i\}=\Sigma$);
\item $\psi:\Sigma\rightarrow\intvl{1}{n}$, an injection.
\end{itemize}

The \emph{domain} of $I$ is the image of $\psi$, that is the set 
$L=\{\psi(\sigma)\mid \sigma\in\Sigma\}$.

The \emph{word representation} of $I$ is the $n$-letter word 
$u_1\ u_2\ldots u_n$ over $\Sigma\cup\{\w\}$
(where $\w\notin\Sigma$),
such that for all $i\in L$, $\psi(u_i)=i$, and for 
$i\in\intvl{1}{n}-L$, $u_i=\w$.
\end{defn}

Two examples of 3DT-instances are given in 
Example~\ref{ex:3DTinstance}. Note that 
such instances can be defined by their word representation and by 
their set of triples $T$. 
The empty 3DT-instance, in which $\Sigma=\emptyset$,
can be written with a sequence of $n$ dots, or 
with the empty word $\emptyinst$.



\begin{example}
In this example, we define two 3DT-instances of span 6, $I=\inst$ and $I'=\instp$: 
$$
\begin{array}{lcll}
I&=&a_1\ c_2\ b_1\ b_2\ c_1\ a_2&\text{ with }T=\{(a_1,b_1,c_1),(a_2,b_2,c_2)\} \\
I'&=&\;\;\w\;\ b_2\ \;\w\; \ c_2\ \;\w\;\ a_2  &\text{ with }T'=\{(a_2,b_2,c_2)\} 
\end{array}$$
\label{ex:3DTinstance}
Here, $I$ has an alphabet of size 6, $\Sigma=\{a_1,b_1,c_1,a_2,b_2,c_2\}$, 
hence $\psi$ is a bijection ($\psi(a_1)=1$, $\psi(c_2)=2$, $\psi(b_1)=3$, etc). 
The second instance, $I'$, has an alphabet of size 3, $\Sigma'=\{a_2,b_2,c_2\}$, 
with $\psi'(b_2)= 2$, $\psi'(c_2)=4$, $\psi'(a_2)=6$.
\end{example}

\begin{property}\label{prop:tailleLT}
Let $I=\inst$ be a 3DT-instance of span $n$ with domain $L$. Then 
\begin{equation*}
\left\lvert\Sigma \right\rvert
=\left\lvert L\right\rvert
=3\left\lvert T\right\rvert \leq n.
\end{equation*}
\end{property}
\begin{proof}
We have $|\Sigma|=|L|$ since $\psi$ is an injection 
with image $L$. The triples of $T$ partition $\Sigma$
so $|\Sigma|=3|T|$, and finally $L\subseteq\intvl1n$ so 
$|L|\leq n$.
\end{proof}

\begin{defn}
Let $I=\inst$ be a 3DT-instance. The injection $\psi$ 
gives a total order over $\Sigma$, written $\pre_I$ 
(or $\pre$, if there is no ambiguity),  defined by
\begin{equation}
\forall \sigma_1,\sigma_2\in\Sigma,
\quad \sigma_1\pre_I \sigma_2
\Leftrightarrow \psi(\sigma_1)<\psi(\sigma_2)
\end{equation}

Two elements $\sigma_1$ and $\sigma_2$ of $\Sigma$ are called 
\emph{consecutive} if there exists no element 
$x\in \Sigma$ such that $\sigma_1\pre_I x\pre_I \sigma_2$. In this 
case, we write $\sigma_1\preCons_I \sigma_2$ 
(or simply $\sigma_1\preCons \sigma_2$).
\end{defn}

An equivalent definition is that $\sigma_1\pre \sigma_2$ if 
$\sigma_1\ \sigma_2$ is a subword of the word representation of
$I$. Also, $\sigma_1\preCons \sigma_2$ if the word representation 
of $I$ contains a factor of the kind 
$\sigma_1\ \w^*\  \sigma_2$ 
(where $\w^*$ represents any sequence of $l\geq 0$ dots).

Using the triples in $T$, we define a successor function over the 
domain $L$:

\begin{defn}
Let $I=\inst$ be a 3DT-instance with domain $L$. The function 
$\succI:L\rightarrow L$ is defined by:
\begin{align*}
\forall (a,b,c)\in T,\quad \psi(a)&\mapsto \psi(b) \\
\psi(b)&\mapsto \psi(c)\\
\psi(c)&\mapsto \psi(a)
\end{align*}
\end{defn}

Function $\succI$ is a bijection, with no fixed-points, and such that 
$\succI\circ\succI\circ\succI$ is the identity over $L$.
In Example~\ref{ex:3DTinstance}, we have:
\begin{equation*}
\succI=\left(
\begin{matrix} 1&2&3&4&5&6 \\  3&6&5&2&1&4\end{matrix}
\right) \text{\  and \ } \succIp=\left(
\begin{matrix} 2&4&6 \\  4&6&2 \end{matrix}
\right).
\end{equation*}


\subsection{3DT-steps}

\begin{defn}
Let $I=\inst$ be a 3DT-instance, and $(a,b,c)$ be a triple of $T$. 
Write $i=\min\{\psi(a),\psi(b),\psi(c)\}$, $j=\succI(i)$, and 
$k=\succI(j)$.
The triple $(a,b,c)\in T$ is \emph{well-ordered} if we have 
$i<j<k$. 
In such a case, we write $\tau[a,b,c,\psi]$ the transposition 
$\tau_{i,j,k}$.
\end{defn}

An equivalent definition is that $(a,b,c)$ is well-ordered iff 
one of $abc$, $bca$, $cab$ is a 
subword of the word representation of $I$. 
In Example~\ref{ex:3DTinstance}, $(a_1,b_1,c_1)$ is 
well-ordered in $I$: indeed, we have $i=\psi(a_1)$, $j=\psi(b_1)$ 
and $k=\psi(c_1)$, so $i<j<k$. The triple $(a_2,b_2,c_2)$ is also 
well-ordered in $I'$ ($i=\psi'(b_2)<j=\psi'(c_2)<k=\psi'(a_2)$), 
but not in $I$: 
$i=\psi(c_2)<k=\psi(b_2)<j=\psi(a_2)$. In this example, we have 
$\tau[a_1,b_1,c_1,\psi]=\tau_{1,3,5}$ and 
$\tau[a_2,b_2,c_2,\psi']=\tau_{2,4,6}$. 


\begin{defn}[3DT-step]
Let $I=\inst$ be a 3DT-instance with $(a,b,c)\in T$ a well-ordered triple.  
The \emph{3DT-step} of parameter $(a,b,c)$ is the operation written $\DT{a,b,c}$,
transforming $I$ into the 3DT-instance $I'=\instp$ such that:
\begin{itemize}
\item $\Sigma'=\Sigma-\{a,b,c\}$
\item $T'=T-\{(a,b,c)\}$
\item $\psi':\begin{array}{ccc}\Sigma' &\rightarrow & \intvl{1}{n}\\\sigma&\mapsto& \tau^{-1}(\psi(\sigma))\end{array}$ $($with $\tau=\tau[a,b,c,\psi])$.
\end{itemize}
\end{defn}

\begin{figure}\begin{center}
\includegraphics{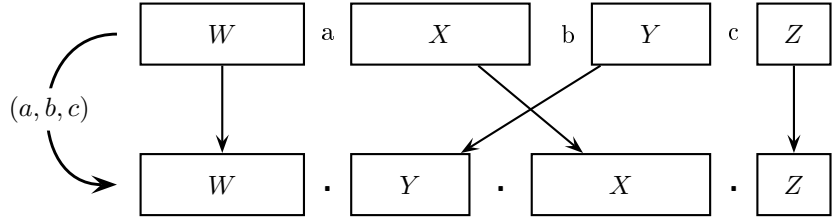}
\caption{\label{fig:3DT_step} The 3DT-step $\DT{a,b,c}$ has two effects, 
here represented on the word representation of a 3DT-instance: 
the triple $(a,b,c)$ is deleted
(and replaced by dots in this word representation), 
and the factors $X$ and $Y$ are swapped.}
\end{center}
\end{figure}

A 3DT-step has two effects on a 3DT-instance, as represented in Figure~\ref{fig:3DT_step}. 
The first is to remove a necessarily
well-ordered triple from $T$ (hence from $\Sigma$). The second is, by applying a 
transposition to $\psi$, to shift the position of some of the remaining elements.
Note that a triple that is not well-ordered in $I$ can become well-ordered in $I'$, 
or vice-versa. In Example~\ref{ex:3DTinstance}, $I'$ can be obtained 
from $I$ via a 3DT-step: $I\DT{a_1,b_1,c_1} I'$. Moreover, $I'\DT{a_2,b_2,c_2} \emptyinst$. A more complex example is given 
in Figure~\ref{fig:example_var}. 

Note that a 3DT-step transforms the function $\succI$ into $\succIp=\tau^{-1}\circ\succI \circ \tau$, 
restricted to~$L'$, the domain of the new instance~$I'$. Indeed, for all $(a,b,c)\in T'$, we have 
\begin{align*}
\succIp(\psi'(a))&=\psi'(b)\\
&=\tau^{-1}(\psi(b))\\
&=\tau^{-1}(\succI(\psi(a)))\\
&=\tau^{-1}(\succI(\tau(\psi'(a))))\\
&=(\tau^{-1}\circ\succI\circ\tau)(\psi'(a))
\end{align*}
The computation is similar for $\psi'(b)$ and $\psi'(c)$.

\begin{figure}\begin{center}
\includegraphics{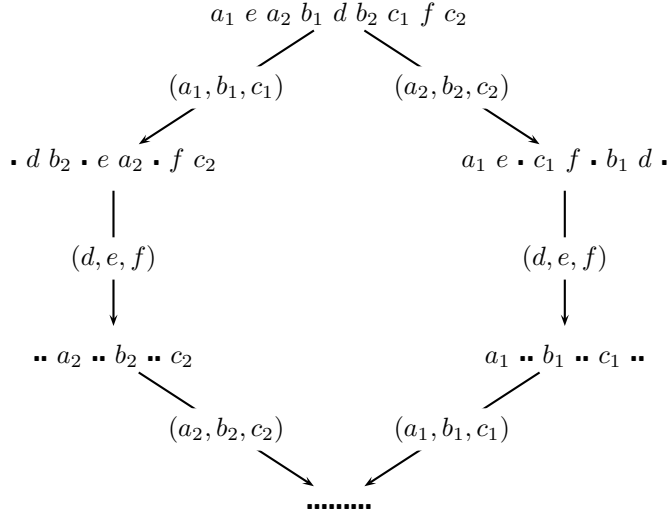}
\caption{\label{fig:example_var} Possible 3DT-steps from the instance $I$ defined 
by the word $a_1\ e\ a_2\ b_1\ d\ b_2\ c_1\ f\ c_2$ and the set of triples 
$T=\{(a_1,b_1,c_1),(a_2,b_2,c_2),(d,e,f)\}$. We can see that there
is a path from $I$ to $\emptyinst$, hence $I$ is 3DT-collapsible.
Note that both $(a_1,b_1,c_1)$ 
and $(a_2,b_2,c_2)$ are well-ordered in the initial instance, each one loses this
property after applying the 3DT-step associated to the other, and becomes well-ordered 
again after applying the 3DT-step associated to $(d,e,f)$. }
\end{center}
\end{figure}

\begin{defn}[3DT-collapsibility]
A 3DT-instance $I=\inst$ is \emph{3DT-collapsible} if there exists a sequence of 3DT-instances $I_k,I_{k-1}, \ldots, I_0$ such that
\begin{align*}
&I_k=I\\
&\forall i\in\intvl1k,\quad \exists(a,b,c)\in T, \quad I_i \DT{a,b,c} I_{i-1}\\
&I_0=\emptyinst
\end{align*}
\end{defn}

In Example~\ref{ex:3DTinstance}, $I$ and $I'$ are 3DT-collapsible, since $I\DT{a_1,b_1,c_1}I'\DT{a_2,b_2,c_2}\emptyinst$.
Another example is the 3DT-instance defined in Figure~\ref{fig:example_var}. 
Note that in the example of Figure~\ref{fig:example_var}, there are in fact two distinct paths leading to
the empty instance.


\subsection{Equivalence with the transposition distance}
\begin{defn}
\label{def:equiv}
Let $I=\inst$ be a 3DT-instance of span $n$ with domain $L$, and $\pi$ be a permutation of $\intvl{0}{n}$.
 We say that $I$ and $\pi$ are equivalent, and we write $I \sim \pi$, if: \addtolength{\arraycolsep}{-4pt}
$$\begin{array}{lrl}
&\pi(0)&=0, \\
\forall v\in \intvl{1}{n}-\L,\quad & \pi(v)&=\pi(v-1)+1, \\
\forall v\in \L,&\pi(v)&=\pi(\succI^{-1}(v)-1)+1.
\end{array}$$\addtolength{\arraycolsep}{4pt}
\end{defn}

With such an equivalence $I\sim \pi$, the two following properties
hold:
\begin{itemize}
\item The breakpoints of $\pi$ correspond to the elements of $L$ (see Property~\ref{prop:3kbreakpoints}).
\item The triples of breakpoints that may be resolved immediately by 
a single transposition correspond to the 
well-ordered triples of $T$
(see Figure~\ref{fig:equiv2} and Lemma~\ref{lem:forceTransposition}). 
\end{itemize}

\begin{figure}\begin{center}
\includegraphics{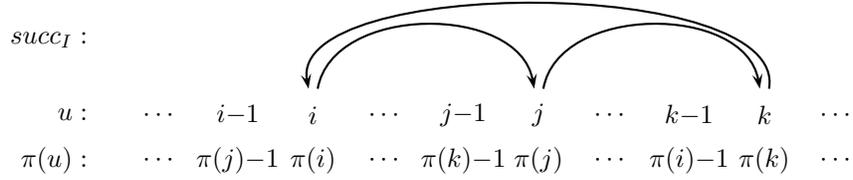}
\caption{\label{fig:equiv2} Illustration of the equivalence $I\sim \pi$ on three integers $(i,j,k)$ such that 
$j=\succI(i)$ and $k=\succI(j)$. It can be checked that $\pi(v)=\pi(u-1)+1$ for 
any $(u,v)\in\{(i,j),(j,k),(k,i)\}$.}
\end{center}
\end{figure}


\begin{property}\label{prop:3kbreakpoints}
Let $I=\inst$ be a 3DT-instance of span $n$ with domain $L$, and $\pi$ be a permutation of $\intvl{0}{n}$, such that $I \sim \pi$.
Then the number of breakpoints of $\pi$ is $d_b(\pi)=\sL=3|T|$.
\end{property}

\begin{proof}
Let $v\in\intvl1n$. By Definition~\ref{def:equiv}, we have:

If $v\notin L$, then $\pi(v)=\pi(v-1)+1$, so $(v-1,v)$ is an adjacency of $\pi$.

If $v\in L$, we write $u=\succI^{-1}(v)$, so $\pi(v)=\pi(u-1)+1$. Since $\succI$ has no fixed-point,
we have $u\neq v$, which implies $\pi(u-1)\neq\pi(v-1)$. Hence, $\pi(v)\neq\pi(v-1)+1$, and $(v-1,v)$ is a 
breakpoint of $\pi$.

Consequently the number of breakpoints of $\pi$ is exactly $\sL$, and $\sL=3|T|$ by Property~\ref{prop:tailleLT}.
\end{proof}

\begin{figure}\begin{center}
\includegraphics{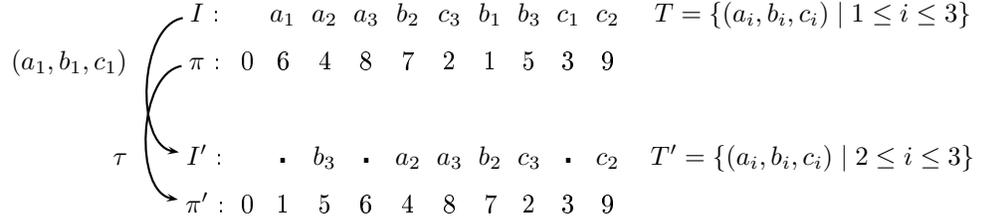}
\caption{\label{fig:ex2c} Illustration of Lemma~\ref{lem:keepEquivalence}: 
since $I \sim \pi$ and $I \DT{a_1,b_1,c_1} I'$, then 
$I' \sim \pi'=\pi\circ \tau$, where $\tau=\tau[a_1,b_1,c_1,\psi]$.}
\end{center}
\end{figure}

With the following lemma, we show that the equivalence between a 
3DT-instance and a permutation is preserved after a 3DT-step, see Figure~\ref{fig:ex2c}.
\begin{lem} \label{lem:keepEquivalence}
Let $I=\inst$ be a 3DT-instance of span $n$, and $\pi$ be 
a permutation of $\intvl{0}{n}$, such that $I \sim \pi$.
If there exists a 3DT-step $I \DT{a,b,c} I'$, then $I'$ 
and $\pi'=\pi\circ \tau$, where $\tau=\tau[a,b,c,\psi]$, are equivalent.
\end{lem}

\begin{proof}
We write $(i,j,k)$ the indices such that $\tau=\tau_{i,j,k}$ 
(i.e. $i=\min\{\psi(a),\psi(b),\psi(c)\}$, $j=\succI(i)$, $k=\succI(j)$). 
Since $(a,b,c)$ is well-ordered, we have $i<j<k$.

We have $I'=\instp$, with $\Sigma'=\Sigma-\{a,b,c\}$, $T'=T-\{(a,b,c)\}$, and $\psi':\sigma\mapsto \tau^{-1}(\psi(\sigma))$. We write respectively $L$ and $L'$ the domains of $I$ and $I'$. For all $v'\in\intvl1n$, we have 
\begin{eqnarray*}
 v'\in L' &\Leftrightarrow& \exists \sigma\in\Sigma-\{a,b,c\},\ v'=\tau^{-1}(\psi(\sigma))\\
 &\Leftrightarrow& \tau(v')\in L-\{i,j,k\}
\end{eqnarray*}
We prove the 3 required properties (see Definition~\ref{def:equiv}) sequentially:
\begin{itemize}
\item $\pi'(0)=\pi(\tau(0))=\pi(0)=0$,

\item $\forall v'\in \intvl{1}{n}-L'$, let $v=\tau(v')$. 
Since $v'\notin L'$, we have either $v\in\{i,j,k\}$, or $v\notin L$.
In the first case, we write $u=\succI^{-1}(v)$ (then $u\in\{i,j,k\}$).  
By Property~\ref{prop:tauXmoinsUn}, $\tau^{-1}(u-1)$ is equal to $\tau^{-1}(\succI(u))-1$, so $\tau^{-1}(u-1)=\tau^{-1}(v)-1$. Hence, 
\begin{align*}
\pi'(v'-1)+1&=\pi(\tau(\tau^{-1}(v)-1))+1\\
&=\pi(u-1)+1\\
&=\pi(v)\text{ by Def.~\ref{def:equiv}, since } v\in L\text{ and }v=\succI(u)\\
&=\pi'(v')
\end{align*}
In the second case, $v\notin\L$, we have 
\begin{align*}
\pi'(v'-1)+1&=\pi(\tau(\tau^{-1}(v)-1))+1\\
&=\pi(\tau(\tau^{-1}(v-1)))+1\text{ by Prop.~\ref{prop:tauXmoinsUn}, since } v\notin\{i,j,k\}\\
&=\pi(v-1)+1\\
&=\pi(v)\text{ by Def.~\ref{def:equiv}, since } v\notin\L\\
&=\pi'(v')
\end{align*}
In both cases, we indeed have $\pi'(v'-1)+1=\pi'(v')$.

\item Let $v'$ be an element of $\Lp$. We write $v=\tau(v')$, $u=\succI^{-1}(v)$, and $u'=\tau^{-1}(u)$.
Then $v'=\tau^{-1}(\succI (\tau(u')))=\succIp(u')$. Moreover, $v\notin\{i,j,k\}$, hence $u\notin\{i,j,k\}$.
\begin{align*}
\pi'(u'-1)+1&=\pi(\tau(\tau^{-1}(u)-1))+1\\
&=\pi(\tau(\tau^{-1}(u-1)))+1\text{ by Prop.~\ref{prop:tauXmoinsUn}, since } u\notin\{i,j,k\}\\
&=\pi(u-1)+1\\
&=\pi(v)\text{ by Def.~\ref{def:equiv}, since } v\in\L\text{ and }u=\succI^{-1}(v)\\
&=\pi(\tau(\tau^{-1}(v)))\\
&=\pi'(v')
\end{align*}
\end{itemize}
\end{proof}

\begin{lem} \label{lem:forceTransposition}
Let $I=\inst$ be a 3DT-instance of span $n$, and $\pi$ a permutation of $\intvl{0}{n}$, such that $I \sim \pi$.
If there exists a transposition $\tau=\tau_{i,j,k}$ such that $d_b(\pi\circ\tau)=d_b(\pi)-3$, then $T$ contains 
a well-ordered triple $(a,b,c)$ such that $\tau=\tau[a,b,c,\psi]$.
\end{lem}

\begin{proof}
We write $i'=\tau^{-1}(i)$, $j'=\tau^{-1}(j)$, and $k'=\tau^{-1}(k)$. Note that $i<j<k$.

Let $\pi'=\pi\circ\tau$. For all $x\in\intvl{1}{n}-\{i,j,k\}$,  we have, by Property~\ref{prop:td-bd}, that
$(x-1,x)$ is an adjacency of $\pi$ iff $(\tau^{-1}(x)-1,\tau^{-1}(x))$ is an adjacency of $\pi'$.
Hence, since $d_b(\pi')=d_b(\pi)-3$, 
we necessarily have that $(i-1,i)$, $(j-1,j)$ and $(k-1,k)$ are 
breakpoints of $\pi$, and $(i'-1,i')$, $(j'-1,j')$ and $(k'-1,k')$ are 
adjacencies of $\pi'$. We have
\begin{align*}
\pi(i)  &=\pi(\tau(i'))\\
	&=\pi'(i')\\
	&=\pi'(i'-1)+1\text{ since $(i'-1,i')$ is an adjacency of $\pi'$}\\
	&=\pi'(\tau^{-1}(i)-1)+1\\
	&=\pi'(\tau^{-1}(k-1))+1\text{ by Prop.~\ref{prop:tauXmoinsUn}}\\
	&=\pi(k-1)+1\\
\end{align*}
Since $I \sim\pi$ and $i\neq k$, by Definition~\ref{def:equiv}, 
we necessarily have $i\in L$ (where $L$ is the domain of $I$), and $i=\succI(k)$.
	
Using the same method with $(j'-1,j')$ and $(k'-1,k')$, we obtain $j,k\in L$,  $j=\succI(i)$ and $k=\succI(j)$. 
Hence, $T$ contains one of the following three triples: $(\psi^{-1}(i),\psi^{-1}(j),\psi^{-1}(k))$, $(\psi^{-1}(j),\psi^{-1}(k),\psi^{-1}(i))$ or $(\psi^{-1}(k),\psi^{-1}(i),\psi^{-1}(j))$. Writing $(a,b,c)$ this triple, we indeed have $\tau_{i,j,k}=\tau[a,b,c,\psi]$ since $i<j<k$.
\end{proof}

\begin{thm}  \label{thm:useEquiv}
Let $I=\inst$ be a 3DT-instance of span $n$ with domain $L$, and $\pi$ be a permutation of $\intvl{0}{n}$, such that $I \sim \pi$.
Then $I$ is 3DT-collapsible if and only if $d_t(\pi)=|T|=d_b(\pi)/3$.
\end{thm}

\begin{proof}
We prove the theorem by induction on $k=|T|$.
For $k=0$, necessarily $I=\emptyinst$ and $L=T=\emptyset$, and by Definition~\ref{def:equiv}, $\pi=Id_n$ 
($\pi(0)=0$, and for all $v>0$, $\pi(v)=\pi(v-1)+1$). 
In this case, $I$ is trivially 3DT-collapsible, and $d_t(\pi)=0=|T|=d_b(\pi)/3$.

Suppose now $k=k'+1$, with $k'\geq 0$, and the theorem is true for $k'$. 
By Property~\ref{prop:3kbreakpoints}, we have $d_b(\pi)=3k$, 
and by Property~\ref{prop:td-bd2}, $d_t(\pi)\geq 3k/3=k$.

Assume first that $I$ is 3DT-collapsible. 
Then there exist both a triple $(a,b,c)\in T$ and 
a 3DT-instance $I'=\instp$ such that 
$I\DT{a,b,c} I'$, and that $I'$ is 3DT-collapsible. 
Since $T'=T-\{(a,b,c)\}$, the size of $T'$ is $k-1=k'$. 
By Lemma~\ref{lem:keepEquivalence}, we have  
$I'\sim \pi'=\pi\circ\tau$, with $\tau=\tau[a,b,c,\psi]$.
Using the induction hypothesis, we know that $d_t(\pi')=k'$. 
So the transposition distance from $\pi=\pi'\circ\tau^{-1}$ 
to the identity is at most, hence exactly, $k'+1=k$.

Assume now that $d_t(\pi)=k$. 
We can decompose $\pi$ into $\pi=\pi'\circ\tau^{-1}$, where $\tau$ 
is a transposition and $\pi'$ a permutation such that $d_t(\pi')=k-1=k'$.
Since $\pi$ has $3k$ breakpoints (Property~\ref{prop:3kbreakpoints}), 
and $\pi'=\pi\circ\tau$ has at most $3k-3$ breakpoints (Property~\ref{prop:td-bd2}), 
$\tau$ necessarily removes $3$ breakpoints, and we can use 
Lemma~\ref{lem:forceTransposition}: there exists a 3DT-step $I\DT{a,b,c}I'$,
where $(a,b,c)\in T$ is a well-ordered triple and $\tau=\tau[a,b,c,\psi]$.
We can now use Lemma~\ref{lem:keepEquivalence}, which yields $I'\sim \pi'=\pi\circ \tau$. 
Using the induction hypothesis, we obtain that 
$I'$ is 3DT-collapsible, hence $I$ is also 3DT-collapsible. This concludes 
the proof of the theorem.
\end{proof}

The previous theorem gives a way to reduce the problem of deciding if a 3DT-instance
is collapsible to the \textsc{Sorting by Transpositions} problem. However, it must be used carefully,
since there exist 3DT-instances to which no permutation is equivalent 
(for example, $I=a_1\ a_2\ b_1\ b_2\ c_1\ c_2$ admits no permutation $\pi$ of $\intvl06$ such that $I\sim\pi$).



\newpage

\section{3DT-collapsibility is NP-Hard to Decide}
\label{sect:3DTisNP-hard}

In this section, we define, for any boolean formula $\phi$, 
a corresponding 3DT-instance $I_\phi$. We also prove that $I_\phi$ is 
3DT-collapsible if and only if $\phi$ is satisfiable. 



\subsection{Block Structure}

The construction of the 3DT-instance $I_\phi$ uses a decomposition 
into blocks, defined below. 
Some triples will be included in a block, in order to define its behavior, 
while others will be shared between two blocks, in order to pass information. 
The former are unconstrained, so that we can design blocks with the behavior 
we need (for example, blocks mimicking usual boolean functions), 
while the latter need to follow several rules, so that the blocks 
can conveniently 
be arranged together.
%
%

\begin{defn}[$l$-block-decomposition]
An \emph{$l$-block-decomposition} $\BB$ of a 3DT-instance $I$ of span $n$ is an $l$-tuple
$(s_1,\ldots,s_l)$ such that $s_1=0$, for all $h\in\intvl{1}{l-1}$, $s_{h}< s_{h+1}$ and $s_l<n$.
We write $t_h=s_{h+1}$ for $h\in\intvl{1}{l-1}$, and $t_l=n$.

Let $I=\inst$. For $h\in\intvl1l$, the factor $u_{s_h+1}\ u_{s_h+2}\ \ldots u_{t_h}$ of the 
word representation $u_1\ u_2\ \ldots u_n$ of $I$ is called
the \emph{full block} $\Bdot h$ (it is a word over $\Sigma\cup\{\w\}$).
The subword of $\Bdot h$ where every occurrence of $\w$ is deleted is called the \emph{block} $\B h$.

For $\sigma\in \Sigma$, we write $\block(\sigma)=h$ if $\psi(\sigma)\in\intvl{s_h+1}{t_h}$ (equivalently, if $\sigma$ appears in the word $\B h$).
A triple $(a,b,c)\in T$ is 
said to be \emph{internal} if $\block(a)=\block(b)=\block(c)$, \emph{external} otherwise.

If $\tau$ is a transposition such that  
for all $h\in\intvl{1}{l}$, $\tau(s_h)<\tau(t_h)$, 
we write $\tau[\BB]$ the $l$-block-decomposition $(\tau(s_1),\ldots,\tau(s_l))$.
\end{defn}

In the rest of this section, we mostly work with blocks instead of full blocks, since we are 
only interested in the relative order of the elements, rather than their actual position. Full
blocks are only used in definitions, where we want to control the dots in the word representation
of the 3DT-instances we define. Note that, for $\sigma_1,\sigma_2\in\Sigma$ 
such that $\block(\sigma_1)=\block(\sigma_2)=h$,
the relation $\sigma_1 \preCons\sigma_2$ is equivalent to \emph{$\sigma_1\ \sigma_2$ is a factor of $\B h$}.

\begin{property}
Let $\BB=(s_1,\ldots,s_l)$ be an $l$-block-decomposition of a 3DT-instance of span $n$, and $i,j,k\in\intvl1n$ be three integers such that 
(a) $i<j<k$ and 
(b) $\exists {h_0}$ such that
$s_{h_0}<i<j\leq t_{h_0}$ or $s_{h_0}<j<k\leq t_{h_0}$ (or both). 
Then for all $h\in\intvl{1}{l}$, $\tau_{i,j,k}^{-1}(s_{h})<\tau_{i,j,k}^{-1}(t_{h})$, 
and the $l$-block-decomposition $\tau_{i,j,k}^{-1}[\BB]$ is defined.
\end{property}
\begin{proof}
For any $h\in\intvl1l$, we show that we cannot have $i\leq s_{h}<j\leq t_{h}<k$. 
Indeed, $s_h<j$ implies $h\leq h_0$ (since $s_h<j\leq t_{h_0}=s_{h_0+1}$),
and $j\leq t_h$ implies $h\geq h_0$ (since $t_{h_0-1}=s_{h_0}<j\leq t_h$). 
Hence $s_{h}<j\leq t_{h}$ implies $h=h_0$, but 
$i\leq s_{h}, t_{h}<k$ contradicts both conditions 
$s_{h_0}<i$ and $k\leq t_{h_0}$: hence
the relation $i\leq s_{h}<j\leq t_{h}<k$ is impossible.

By Property~\ref{prop:echangeTau}, since $s_h<t_h$ for all $h\in \intvl1l$, 
and  $i\leq s_{h}<j\leq t_{h}<k$ does not hold, 
we have $\tau_{i,j,k}^{-1}(s_{h})<\tau_{i,j,k}^{-1}(t_{h})$,
which is sufficient to define $\tau_{i,j,k}^{-1}[\BB]$.
\end{proof}

The above property yields that, if $(a,b,c)$ is a well-ordered
triple of a 3DT-instance $I=\inst$ ($\tau=\tau[a,b,c,\psi]$), 
and $\BB$ is an $l$-block-decomposition of $I$,
then $\tau^{-1}[\BB]$ is defined if
$(a,b,c)$ is an internal triple, 
or an external triple such that one of the following equalities is satisfied:
$\block(a)=\block(b)$, $\block(b)=\block(c)$ or $\block(c)=\block(a)$.
In this case, the 3DT-step $I\DT{a,b,c}I'$ is written 
$(I,\BB)\DT{a,b,c}(I',\BB')$, where $\BB'=\tau^{-1}[\BB]$ is an $l$-block-decomposition of $I'$.

\begin{defn}[Variable]\label{def:validVar}
A \emph{variable} $A$ of a 3DT-instance $I=\inst$ is a pair of triples $A=[(a,b,c),(x,y,z)]$ of $T$.
It is \emph{valid} in an $l$-block-decomposition $\BB$ if 
\begin{itemize}
\item[\Cii]  $\exists h_0\in\intvl1l$ such that $\block(b)=\block(x)=\block(y) =h_0$ 
\item[\Ciii]  $\exists h_1\in\intvl1l$, $h_1\neq h_0$, such that $\block(a)=\block(c)=\block(z) =h_1$ 
\item[\Civ]  if $x\pre y$, then we have $x\preCons b\preCons y$
\item[\Cv]   $a\pre z\pre c$
\end{itemize}

For such a valid variable $A$, 
the block $\B {h_0}$ containing $\{b,x,y\}$ is called the \emph{source} of $A$ (we write $source(A)={h_0}$), 
and the block $\B {h_1}$ containing $\{a,c,z\}$ is called the \emph{target} of $A$ (we write $target(A)={h_1}$). 
For $h\in\intvl1l$, the variables of which $\B h$ is the source (resp. the target) 
are called the \emph{output} (resp. the \emph{input}) of $\B h$.
The 3DT-step $I\DT{x,y,z}I'$ is called the \emph{activation} of $A$
(it requires that $(x,y,z)$ is well-ordered).
\end{defn}

Note that since a valid variable $A=[(a,b,c),(x,y,z)]$ has $\block(x)=\block(y)$,
its activation can be written $(I,\BB)\DT{x,y,z}(I',\BB')$. 

\begin{property} \label{prop:onlyXYZ}
Let $(I,\BB)$ be a 3DT-instance with an $l$-block-decomposition,
and $A$ be a variable of $I$ that is valid in $\BB$. Write $A=[(a,b,c), (x,y,z)]$.
Then $(x,y,z)$ is well-ordered iff $x\pre y$; and $(a,b,c)$ is not well-ordered.
\end{property}

\begin{proof}
Note that for all $\sigma,\sigma'\in \Sigma$, $\block(\sigma)<\block(\sigma')\Rightarrow \sigma\pre\sigma'$.
Write $I=\inst$, ${h_0}=source(A)$ and ${h_1}=target(A)$: we have $h_0\neq h_1$ by condition \Ciii\ of Definition~\ref{def:validVar}.

If ${h_0}<{h_1}$, then, with condition~\Cv\ of Definition~\ref{def:validVar}, $b\pre a\pre c$, and either $x\pre y\pre z$ or $y\pre x\pre z$. 
Hence, $(a,b,c)$ is not well-ordered, and $(x,y,z)$ is well-ordered iff $x\pre y$.

Likewise, if ${h_1}<{h_0}$, we have $a\pre c\pre b$, and $z\pre x\pre y$ or $z\pre y\pre x$.
Again, $(a,b,c)$ is not well-ordered, and $(x,y,z)$ is well-ordered iff $x\pre y$.
\end{proof}

\begin{property}\label{prop:varEffects}
Let $(I,\BB)$ be a 3DT-instance with an $l$-block-decomposition,
such that the external triples of $I=\inst$ can be partitioned into 
a set of valid variables $\mathcal A$. 
Let $(d,e,f)$ be a well-ordered triple of $I$, such that there exists
a 3DT-step $(I,\BB)\DT{d,e,f}(I',\BB')$, with $I'=\instp$. Then one of the two following 
cases is true:

\begin{itemize}
\item $(d,e,f)$ is an internal triple. We write $h_0=\block(d)=\block(e)=\block(f)$.
Then for all $\sigma\in\Sigma'$, $\blockIp(\sigma)=\blockI(\sigma)$.
Moreover if $\sigma_1,\sigma_2\in \Sigma'$ with $ \blockIp(\sigma_1)= \blockIp(\sigma_2)\neq h_0$ and $\sigma_1\pre_I\sigma_2$, 
then $\sigma_1\pre_{I'}\sigma_2$.

\item $\exists A=[(a,b,c),(x,y,z)]\in \mathcal A$, such that $(d,e,f)=(x,y,z)$. 
Then $\blockIp(b)=target(A)$ and for all $\sigma\in\Sigma'-\{b\}$, $\blockIp(\sigma)=\blockI(\sigma)$.
Moreover if $\sigma_1,\sigma_2\in \Sigma'-\{b\}$, such that $\sigma_1\pre_I\sigma_2$, then 
$\sigma_1\pre_{I'}\sigma_2$.
\end{itemize}
\end{property}

\begin{proof}
We respectively write $\tau$ and $i,j,k$ 
the transposition and the three integers such that 
$\tau=\tau_{i,j,k}=\tau[d,e,f,\psi]$ (necessarily, $0<i<j<k\leq n$). We also 
write $\BB=(s_0,s_1,\ldots,s_l)$.
The triple $(d,e,f)$ is either internal or external in $\BB$. 

If $(d,e,f)$ is internal, with $h_0=\block(d)=\block(e)=\block(f)$, we have (see Figure~\ref{fig:blocks_after}a): \begin{equation*}
{s_{h_0}<i<j<k\leq t_{h_0}}.\end{equation*}

Hence for all $h\in\intvl1l$, either $s_h<i$ or $k\leq s_h$, and $\tau^{-1}(s_h)=s_h$ by Definition~\ref{def:tau}. 
Moreover, for all $\sigma\in\Sigma$, we have
\begin{eqnarray*}
i\leq \psi(\sigma)<k &\Rightarrow & \psi(\sigma)\in\intvl{s_{h_0}+1}{t_{h_0}}\text{ and }\tau^{-1}(s_{h_0})< i\leq \tau^{-1}(\psi(\sigma))<k\leq \tau^{-1}(t_{h_0})\\
&\Rightarrow & \block(\sigma)=h_0= \blockIp(\sigma)\\
\psi(\sigma)<i\text{ or } k\leq \psi(\sigma) &\Rightarrow & \tau^{-1}(\psi(\sigma))=\psi(\sigma)\\
&\Rightarrow & \blockIp(\sigma)=\blockI(\sigma)
\end{eqnarray*}
Finally, if $\sigma_1,\sigma_2\in \Sigma'$ with $ \blockIp(\sigma_1)= \blockIp(\sigma_2)\neq h_0$, then 
we have both $\tau^{-1}(\psi(\sigma_1))=\psi(\sigma_1)$ and $\tau^{-1}(\psi(\sigma_2))=\psi(\sigma_2)$.
Hence $\sigma_1\pre_I\sigma_2 \Leftrightarrow \sigma_1\pre_{I'}\sigma_2$.

If $(d,e,f)$ is external, then, since the set of external triples can be partitioned into variables,
there exists a variable $A=[(a,b,c),(x,y,z)]\in \mathcal A$, such that $(d,e,f)=(a,b,c)$ or $(d,e,f)=(x,y,z)$. 
Since $(d,e,f)$ is well-ordered in $I$, we have, by Property~\ref{prop:onlyXYZ}, $(d,e,f)=(x,y,z)$ and $x\pre_I y$, 
see Figure~\ref{fig:blocks_after}b. 
And since $A$ is valid, by condition 
\Cv\ of Definition~\ref{def:validVar}, $x\preCons_I b\preCons_I y$.
We write $h_0=source(A)$ and $h_1=target(A)$, and we assume that $h_0<h_1$, 
which implies $x\pre_I y\pre_I z$ 
(the case $h_1<h_0$ with $z\pre_I x\pre_I y$ is similar): 
thus, we have 
\begin{equation*}
i=\psi(x),\ j=\psi(y),\ k=\psi(z),\text{ and }s_{h_0}<i<j\leq t_{h_0}\leq s_{h_1} <k\leq t_{h_1}.
\end{equation*}
%
We define a set of indices $U$ by 
\begin{equation*}
U=\{s_h\mid h\in\intvl1l\}\,\cup\,\{n\}\,\cup\,\{\psi(\sigma)\mid \sigma\in \Sigma'-\{b\}\}.
\end{equation*}

We now show that for all $u\in U$, we have $u< i$ or $j\leq u$. Indeed, if $u=s_h$ for some $h\in \intvl1l$, 
then either $h\leq h_0$ and $u\leq s_{h_0}<i$, or  $h_0<h$ and $j\leq t_{h_0} \leq u$. 
Also, if $u=n$, then $j\leq u$. Finally, assume $u=\psi(\sigma)$, 
with $\sigma\in \Sigma'-\{b\}$.
We then have $x\pre_I\sigma \pre_I y\Leftrightarrow \sigma=b$, since 
$x\preCons_I b \preCons_I y$. 
Hence  either $\sigma\pre_I x$ and $u<\psi(x)=i$, or $y\pre_I \sigma$ and $\psi(y)=j<u$. 

By Property~\ref{prop:echangeTau}, if $u,v\in U$ are such that $u<v$, then $\tau^{-1}(u)<\tau^{-1}(v)$. 
This implies that elements of $\Sigma'-\{b\}=\Sigma-\{b,x,y,z\}$ 
do not change blocks after applying $\tau^{-1}$ on $\psi$,
and that the relative order of any two elements is preserved.
Finally, for $b$, we have $x\pre_I b\pre_I y$, hence
\begin{equation*}
i\leq\psi(b)< j \leq s_{h_1}< k\leq t_{h_1}.
\end{equation*}
Thus, by Property~\ref{prop:echangeTau}, $\tau^{-1}(s_{h_1})<\tau^{-1}(\psi(b))< \tau^{-1}(t_{h_1})$, 
and $\blockIp(b)=h_1=target(A)$. 
This completes the proof.
\end{proof}

\begin{defn}[Valid context]
\label{def:validContext}
A 3DT-instance with an $l$-block-decomposition $(I,\BB)$ 
is a \emph{valid context} if the set of external triples of 
$I$ can be partitioned into valid variables.
\end{defn}

With the following property, we ensure that a valid context remains 
\emph{almost} valid after applying a 3DT-step: the partition of the external triples into variables if kept through this 3DT-step,
but conditions \Civ\ and \Cv\ of Definition~\ref{def:validVar} are not necessarily satisfied.

\begin{property}\label{prop:keepCiiCiii}
Let $(I,\BB)$ be a valid context 
and $(I,\BB)\DT{d,e,f}(I',\BB')$ be a 3DT-step. 
Then the external triples of $(I',\BB')$ can be partitioned into a set of variables,
each satisfying conditions \Cii\ and \Ciii\ of Definition~\ref{def:validVar}.
\end{property}
\begin{proof}
Let $I=\inst$, $I'=\instp$, $\mathcal A$ be the set of variables of $I$, and
$E$ (resp. $E'$) be the set of external triples
of $I$ (resp. $I'$).
From Property~\ref{prop:varEffects}, two cases are possible.

First case: $(d,e,f)\notin E$.
Then for all $\sigma\in\Sigma'$, $\blockIp(\sigma)=\blockI(\sigma)$.
Hence $E'=E$, and $(I',\BB')$ has the same set of variables as $(I,\BB)$, that is $\mathcal A$.
The source and target blocks of every variable remain unchanged, 
hence conditions \Cii\ and \Ciii\ of Definition~\ref{def:validVar} are still 
satisfied for each $A\in\mathcal A$ in $\BB'$.

Second case: $(d,e,f)\in E$, and $\exists A=[(a,b,c),(x,y,z)]\in \mathcal A$, 
such that $(d,e,f)=(x,y,z)$, by Property~\ref{prop:varEffects}. 
Then $\blockIp(b)=target(A)$
and for all $\sigma\in\Sigma'-\{b\}$, $\blockIp(\sigma)=\blockI(\sigma)$.
Hence $\blockIp(b)=\blockIp(a)=\blockIp(c)$, and $E'=E-\{(x,y,z),(a,b,c)\}$: indeed, $(x,y,z)$ is deleted in $T'$ 
so $(x,y,z)\notin E'$, $(a,b,c)$ is internal in $I'$, and every other triple is untouched.
And for every $A'=[(a',b',c'),(x',d',e')]\in \mathcal A- \{A\}$, we have 
$\blockIp(\sigma)=\blockI(\sigma)$ for $\sigma\in\{a',b',c',x',y',z'\}$, 
hence $A'$ satisfies conditions \Cii\ and \Ciii\ of Definition~\ref{def:validVar} in $\BB'$.
\end{proof}

\begin{figure}\begin{center}
\includegraphics{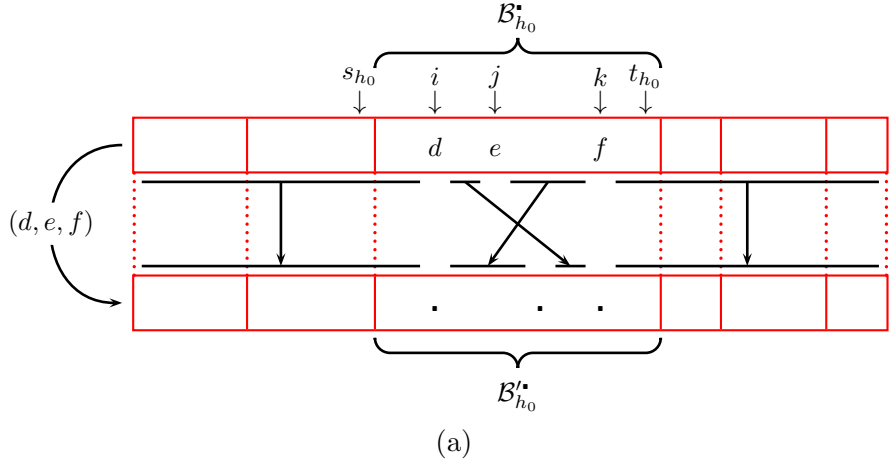}\\(a)\bigskip

\includegraphics{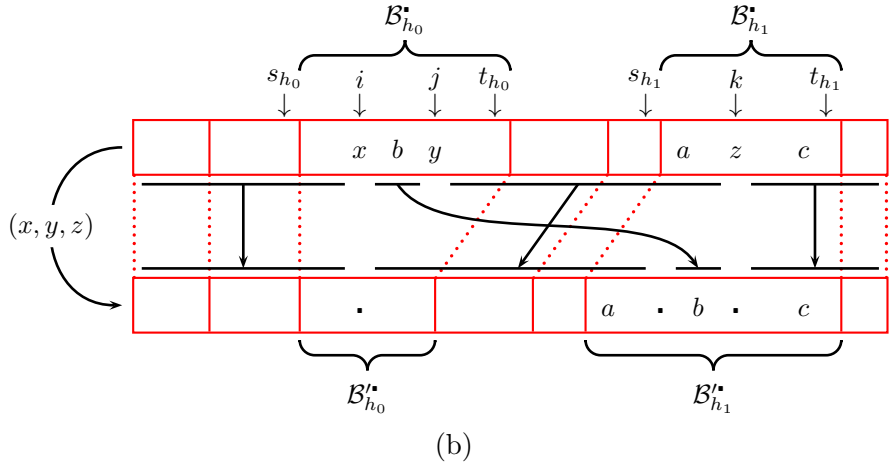}\\(b)
\caption{\label{fig:blocks_after} Effects of a 3DT-step $\DT{d,e,f}$ on an $l$-block-decomposition if 
(a) 
$(d,e,f)$ is an internal triple, 
or (b) 
there exists a variable $A=[(a,b,c),(x,y,z)]$ such that $(d,e,f)=(x,y,z)$.
Both figures are in fact derived from Figure~\ref{fig:3DT_step} in the context of an $l$-block-decomposition.
}
\end{center}
\end{figure}

Consider a block $B$ in a valid context $(I,\BB)$ 
(there exists $h\in\intvl1l$ such that $B=\B h$), 
and $(d,e,f)$ a triple of $I$ such that $(I,\BB)\DT{d,e,f}(I',\BB')$ 
(we write $B'=\Bp h$). Then, 
following Property~\ref{prop:varEffects}, four cases are possible:
\begin{itemize}
\item $h\notin\{\block(d),\block(e),\block(f)\}$, 
hence $B'=B$, since, by Property~\ref{prop:varEffects}, 
the relative order of the elements of $B$ remains unchanged 
after the 3DT-step $\DT{d,e,f}$.
\item $(d,e,f)$ is an internal triple of $B$. We write
\begin{center}\includegraphics{\figurePath{mini_internal}}\end{center}
\item $\exists A=[(a,b,c),(x,y,z)]$ such that $h=source(A)$ and $(d,e,f)=(x,y,z)$ 
($A$ is an output of~$B$), see Figure~\ref{fig:activatingOneVar} (left). We write
\begin{center}\includegraphics{\figurePath{mini_output}}\end{center}
\item $\exists A=[(a,b,c),(x,y,z)]$ such that $h=target(A)$ and $(d,e,f)=(x,y,z)$ 
($A$ is an input of~$B$), see Figure~\ref{fig:activatingOneVar} (right). We write 
\begin{center}\includegraphics{\figurePath{mini_input}}\end{center}
\end{itemize}

The graph obtained from a block $B$ by following exhaustively the 
possible arcs as defined above (always assuming this block is in a valid 
context) is called the \emph{behavior graph} of $B$.


\begin{figure}\begin{center}
\includegraphics{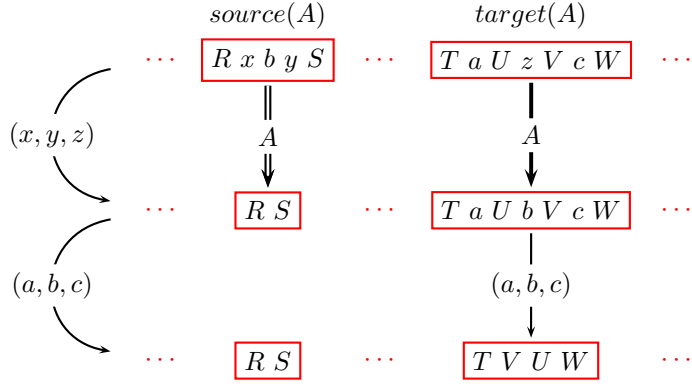}
\caption{\label{fig:activatingOneVar} The activation 
of a variable $A=[(a,b,c),(x,y,z)]$ is written 
with a double arc in the behavior graph of the source block of $A$ and with a thick arc 
in the behavior graph of its target block.
It can be followed by the 3DT-step $\DT{a,b,c}$, impacting only the target block of $A$. 
Dot symbols ($\w$) are omitted. 
We denote by $R,S,T,U,V,W$ some factors of the source and target blocks of $A$: the consequence of activating 
$A$ is to allow $U$ and $V$ to be swapped in $target(A)$.}
\end{center}
\end{figure}



\newpage
\subsection{Basic Blocks}
We now define four basic blocks: 
$\blockFont{copy}$, $\blockFont{and}$, $\blockFont{or}$, and $\blockFont{var}$. 
They are studied independently in this section, before being assembled in Section~\ref{sect:assemblage}.
Each of these blocks is defined by a word
and a set of triples. 
We distinguish internal triples, for which all three elements appear in a single 
block, from external triples, which are part of an input/output variable, and for which 
only one or two elements appear in the block. 
Note that each external triple is part of an input (resp. output) variable,
which itself must be an output (resp. input) of another block, the other block containing
the remaining elements of the triple.

We then draw the behavior graph of each of these blocks 
(Figures~\ref{fig:block_copy2} to~\ref{fig:block_var}): 
in each case, we assume that the block is in a valid context, and 
follow exhaustively the 3DT-steps that can be applied on it. 
We then give another graph (Figures~\ref{fig:blockAbs}a to ~\ref{fig:blockAbs}d), 
obtained from the behavior graph by contracting all arcs 
corresponding to 3DT-steps using
internal triples, i.e. we assimilate every pair of nodes linked by such an arc.
Hence, only the arcs corresponding to the activation of an input/output 
variable remain. From this second figure, we derive a property describing the behavior
of the block, in terms of activating input and output variables 
(always provided this block is in a valid context).
It must be kept in mind that for any variable, it is the state of 
the source block which determines whether it can be activated, 
whereas the activation itself affects mostly the target block.


\subsubsection{The block $\blockFont{copy}$}
This block aims at duplicating a variable: any of the two output variables can only be activated after 
the input variable has been activated.

\noindent
Input variable: $A=[(a,b,c),(x,y,z)]$.

\noindent
Output variables: $A_1=[(a_1,b_1,c_1),(x_1,y_1,z_1)]$ and $A_2=[(a_2,b_2,c_2),(x_2,y_2,z_2)]$.

\noindent
Internal triple: $(d,e,f)$.

\noindent
Definition:
\begin{equation*}
\blockFrame{[A_1,A_2]=\blockFont{copy}(A)} \quad = \quad \blockFrame{a\ y_1\ e\ z\ d\ y_2\ x_1\ b_1\ c\ x_2\ b_2\ f}
\end{equation*}

\begin{property}\label{prop:block_copy2}
In a block $[A_1,A_2]=\blockFont{copy}(A)$ in a valid context, the possible orders in which
$A$, $A_1$ and $A_2$ can be activated are $(A,A_1,A_2)$ and $(A,A_2,A_1)$.
\end{property}
\begin{proof}
See Figures~\ref{fig:block_copy2} and~\ref{fig:blockAbs}a.
\end{proof}

\begin{figure}\begin{center}
\includegraphics{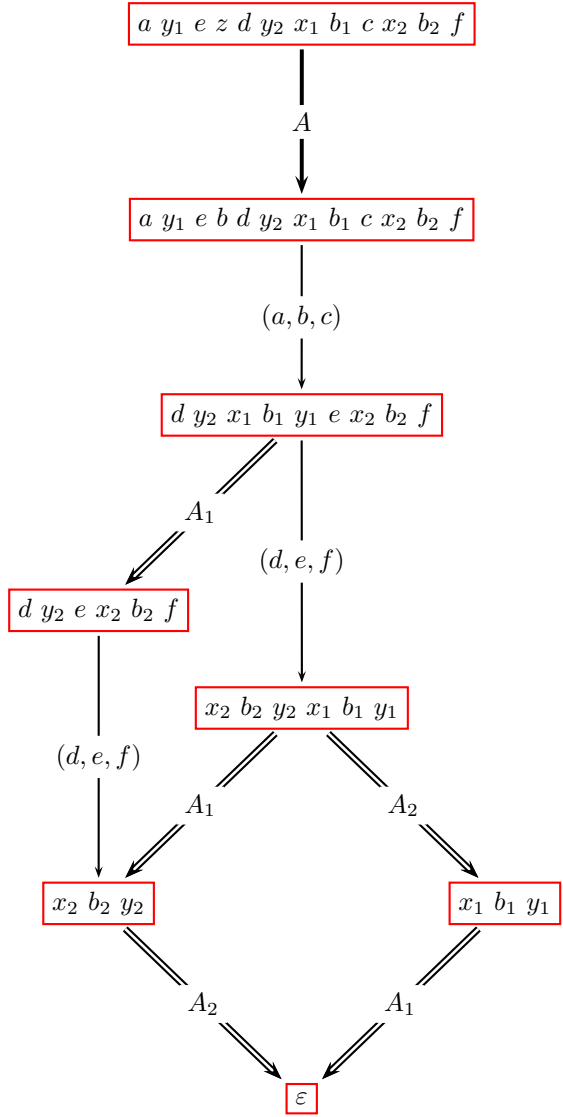}
\caption{\label{fig:block_copy2} Behavior graph of the block $[A_1,A_2]=\blockFont{copy}(A)$.
A thick (resp. double) arc corresponds to the 3DT-step $\DT{x,y,z}$ for an input (resp. output) variable $[(a,b,c),(x,y,z)]$.}
\end{center}
\end{figure}

\subsubsection{The block $\blockFont{and}$}
This block aims at simulating a conjunction: the output variable can only be activated after 
both input variables have been activated.

\noindent
Input variables: $A_1=[(a_1,b_1,c_1),(x_1,y_1,z_1)]$ and $A_2=[(a_2,b_2,c_2),(x_2,y_2,z_2)]$.

\noindent
Output variable: $A=[(a,b,c),(x,y,z)]$.

\noindent
Internal triple: $(d,e,f)$.

\noindent
Definition:
\begin{equation*}
\blockFrame{A=\blockFont{and}(A_1,A_2)}\quad=\quad \blockFrame{a_1\ e\ z_1\ a_2\ c_1\ z_2\ d\ y\ c_2\ x\ b\ f}
\end{equation*}

\begin{property}\label{prop:block_and}
In a block $A=\blockFont{and}(A_1,A_2)$ in a valid context, the possible orders in which
$A$, $A_1$ and $A_2$ can be activated are $(A_1,A_2,A)$ and $(A_2,A_1,A)$.
\end{property}
\begin{proof}
See Figures~\ref{fig:block_and} and~\ref{fig:blockAbs}b.
\end{proof}


\begin{figure}\begin{center}
\includegraphics{\figurePath{block_and2}}
\caption{\label{fig:block_and} Behavior graph of the block $A=\blockFont{and}(A_1,A_2)$.}
\end{center}
\end{figure}

\subsubsection{The block $\blockFont{or}$}
This block aims at simulating a disjunction: the output variable can be activated as soon as any 
of the two input variables is activated.

\noindent
Input variables: $A_1=[(a_1,b_1,c_1),(x_1,y_1,z_1)]$ and $A_2=[(a_2,b_2,c_2),(x_2,y_2,z_2)]$.

\noindent
Output variable: $A=[(a,b,c),(x,y,z)]$.

\noindent
Internal triples: $(a',b',c')$ and $(d,e,f)$.

\noindent
Definition:
\begin{equation*}
\blockFrame{A=\blockFont{or}(A_1,A_2)}\quad=\quad \blockFrame{a_1\ b'\ z_1\ a_2\ d\ y\ a'\ x\ b\ f\ z_2\ c_1\ e\ c'\ c_2}
\end{equation*}

\begin{property}\label{prop:block_or}
In a block $A=\blockFont{or}(A_1,A_2)$ in a valid context, the possible orders in which
$A$, $A_1$ and~$A_2$ can be activated are $(A_1,A,A_2)$, $(A_2,A,A_1)$, $(A_1,A_2,A)$ and $(A_2,A_1,A)$.
\end{property}
\begin{proof}
See Figures~\ref{fig:block_or} and~\ref{fig:blockAbs}c.
\end{proof}

\begin{figure}\begin{center}
\makebox[1pt][c]{
\includegraphics{\figurePath{block_or4}}}
\caption{\label{fig:block_or} Behavior graph of the block $A=\blockFont{or}(A_1,A_2)$.}
\end{center}
\end{figure}

\subsubsection{The block $\blockFont{var}$}

This block aims at simulating a boolean variable: in a first stage, 
only one of the two output variables can be activated. 
The other needs the activation of the input variable to be activated. 

\noindent
Input variable: $A\zz=[(a\zz,b\zz,c\zz),(x\zz,y\zz,z\zz)]$.

\noindent
Output variables: $A_1=[(a_1,b_1,c_1),(x_1,y_1,z_1)]$, $A_2=[(a_2,b_2,c_2),(x_2,y_2,z_2)]$.

\noindent
Internal triples: $(d_1,e_1,f_1)$, $(d_2,e_2,f_2)$ and $(a',b',c')$.

\noindent
Definition:
\begin{equation*}
\blockFrame{[A_1,A_2]=\blockFont{var}(A\zz)}\quad=\quad \blockFrame{d_1\ y_1\ a\zz\ d_2\ y_2\ e_1\ a'\ e_2\ x_1\ b_1\ f_1\ c'\ z\zz\ b'\ c\zz\ x_2\ b_2\ f_2}
\end{equation*}


\begin{property}\label{prop:block_var}
In a block $[A_1,A_2]=\blockFont{var}(A)$ in a valid context, the possible orders in which
$A$, $A_1$ and $A_2$ can be activated are $(A_1,A,A_2)$, $(A_2,A,A_1)$, $(A,A_1,A_2)$ and $(A,A_2,A_1)$.
\end{property}
\begin{proof}
See Figures~\ref{fig:block_var} and~\ref{fig:blockAbs}d.
\end{proof}
With such a block, if $A$ is not activated first, 
one needs to make a choice between activating 
$A_1$ or $A_2$. Once $A$ is activated, however, all remaining output variables are activable.


\begin{figure}\begin{center}
\makebox[1pt][c]{
\includegraphics{\figurePath{block_var3}}}
\caption{\label{fig:block_var} Behavior graph of the block $[A_1,A_2]=\blockFont{var}(A)$.}
\end{center}
\end{figure}

\begin{figure}
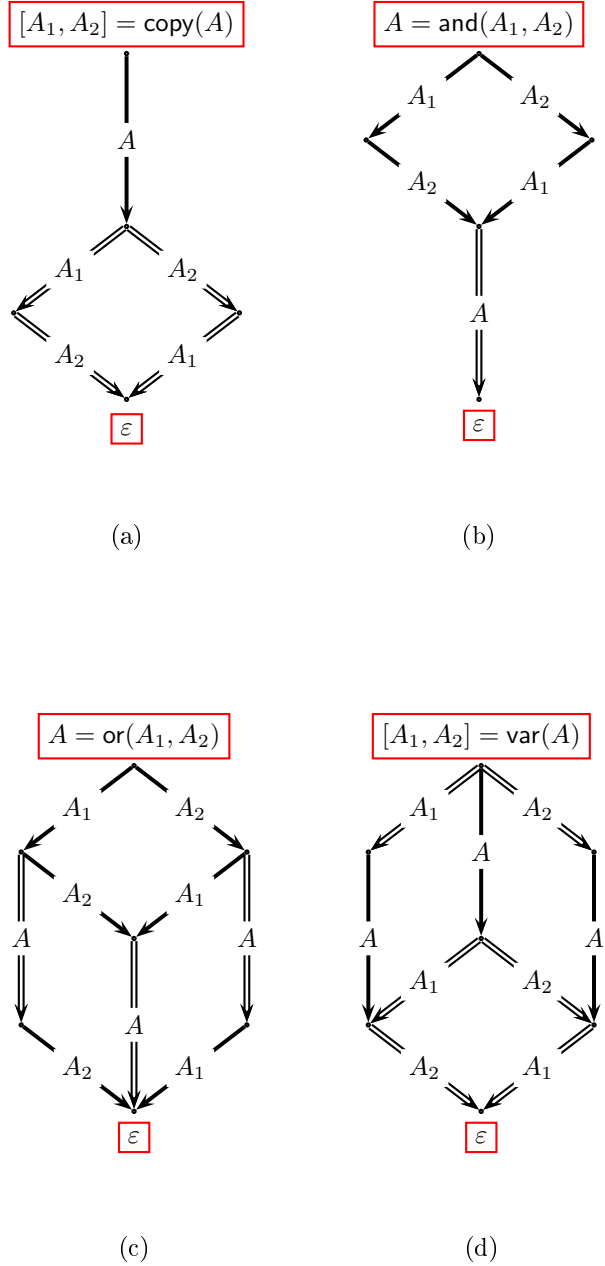
\begin{center}
\includegraphics{\figurePath{blockAbs_copy2}}\qquad
\includegraphics{\figurePath{blockAbs_and}}
\vspace{5em}

\includegraphics{\figurePath{blockAbs_or}}\qquad
\includegraphics{\figurePath{blockAbs_var}}
\caption{\label{fig:blockAbs} Abstract representations of the blocks $\blockFont{copy}$, 
$\blockFont{and}$, $\blockFont{or}$, and $\blockFont{var}$, obtained
from each behavior graph (Figures~\ref{fig:block_copy2}, \ref{fig:block_and}, \ref{fig:block_or} and \ref{fig:block_var})
by contracting arcs corresponding to internal triples, 
and keeping only the arcs corresponding to variables. 
We see, for each block, which output variables are activable, 
depending on which variables have been activated.}
\end{center}
\end{figure}


\subsubsection{Assembling the blocks $\blockFont{copy}$, $\blockFont{and}$, $\blockFont{or}$, $\blockFont{var}$.}
\begin{defn}[Assembling of basic blocks]
An {\em assembling of basic blocks} $(I,\BB)$ is composed of a 3DT-instance $I$ and an $l$-block-decomposition 
$\BB$ obtained by the following process:
\begin{itemize}
\item Create a set of variables $\mathcal A$.
\item Define $I=\inst$ by its word representation, as a concatenation of $l$ factors $\BB^{\w}_1\ \BB^{\w}_2\ \ldots\ \BB^{\w}_l$ and a set of triples~$T$,
where each $\BB^{\w}_h$ is one of the blocks  $[A_1,A_2]=\blockFont{copy}(A)$,  $A=\blockFont{and}(A_1,A_2)$, ${A=\blockFont{or}(A_1,A_2)}$ or $[A_1,A_2]=\blockFont{var}(A)$, with $A_1,A_2,A\in \mathcal A$ (such that each $X\in \mathcal A$ appears in the input of exactly one block, and in the output of exactly one other block); and where $T$ is the union of the set of internal triples needed in each block, and the set of external triples defined by the variables of~$\mathcal A$.
\end{itemize}
\end{defn}


\begin{example}\label{ex:ToyAssembling}
We create a 3DT-instance $I$ with a 2-block-decomposition $\BB$ 
such that $(I,\BB)$ is an assembling of basic blocks, defined as follows:
\begin{itemize}
\item $I$ uses three variables, $\mathcal A=\{X_1,X_2, Y\}$
\item the word representation of $I$ is the concatenation of $[X_1,X_2]=\blockFont{var}(Y)$ and $Y=\blockFont{or}(X_1,X_2)$
\end{itemize}

With $X_1=[(a_1,b_1,c_1),(x_1,y_1,z_1)]$, $X_2=[(a_2,b_2,c_2),(x_2,y_2,z_2)]$, $Y=[(a,b,c),(x,y,z)]$, and the internal triples
$(d_1,e_1,f_1), (d_2,e_2,f_2), (a',b',c')$ for the block 
$\blockFont{var}$, and $(a'',b'',c''), (d,e,f)$ for the block $\blockFont{or}$, the word representation of 
$I$ is the following (note that its 2-block-decomposition is $(0,18)$):
\begin{equation*}
I=d_1\ y_1\ a\ d_2\ y_2\ e_1\ a'\ e_2\ x_1\ b_1\ f_1\ c'\ z\ b'\ c\ x_2\ b_2\ f_2 \quad
a_1\ b''\ z_1\ a_2\ d\ y\ a''\ x\ b\ f\ z_2\ c_1\ e\ c''\ c_2
\end{equation*}
Indeed, a possible sequence of 3DT-steps leading from $I$ to $\emptyinst$ 
is given in Figure~\ref{fig:exToyAssembling}, hence $I$ is 3DT-collapsible. 
\end{example}
\newcommand{\mynote}[2]{
\raisebox{0.8em}{\ }\raisebox{-0.5em}{\ }\hspace{1.9em}\makebox[0pt]{\speclab{#1}}\hspace{39.1em}\makebox[0pt][r]{ \small #2 }
}

\newcommand{\speclab}[1]{\small $#1$}
\newcommand{\spec}[1]{\text{\boldmath$#1$}}

\newcommand{\myunderline}[1]{\underline{\vphantom{f_{f_f}}#1}}
\begin{figure}
\begin{align*}
I&= |\ \spec{d_1}\ \myunderline{y_1\ a\ d_2\ y_2}\ \spec{e_1}\ \myunderline{\mbox{$a'\ e_2\ x_1\ b_1$}}\ \spec{f_1}\ c'\ z\ b'\ c\ x_2\ b_2\ f_2 
    \ |\ a_1\ b''\ z_1\ a_2\ d\ y\ a''\ x\ b\ f\ z_2\ c_1\ e\ c''\ c_2\ | \\
\downarrow &\mynote{(d_1,e_1,f_1)}{Internal triple of $\B1$}\\
I_{10}&= |\ a'\ e_2\ \spec{x_1}\ \myunderline{b_1}\ \spec{y_1}\ \myunderline{a\ d_2\ y_2\ c'\ z\ b'\ c\ x_2\ b_2\ f_2 
    \ |\ a_1\ b''}\ \spec{z_1}\ a_2\ d\ y\ a''\ x\ b\ f\ z_2\ c_1\ e\ c''\ c_2\ | \\
\downarrow &\mynote{(x_1,y_1,z_1)}{\bf Activation of $X_1$}\\
I_{9}&= |\ a'\ e_2\ a\ d_2\ y_2\ c'\ z\ b'\ c\ x_2\ b_2\ f_2 
    \ |\ \spec{a_1}\ \myunderline{b''}\ \spec{b_1}\ \myunderline{a_2\ d\ y\ a''\ x\ b\ f\ z_2}\ \spec{c_1}\ e\ c''\ c_2\ | \\
\downarrow &\mynote{(a_1,b_1,c_1)}{Internal triple of $\B2$}\\
I_{8}&= |\ a'\ e_2\ a\ d_2\ y_2\ c'\ z\ b'\ c\ x_2\ b_2\ f_2 
    \ |\ a_2\ d\ y\ \spec{a''}\ \myunderline{x\ b\ f\ z_2}\ \spec{b''}\ \myunderline{e}\ \spec{c''}\ c_2\ | \\
\downarrow &\mynote{(a'',b'',c'')}{Internal triple of $\B2$}\\
I_{7}&= |\ a'\ e_2\ a\ d_2\ y_2\ c'\ z\ b'\ c\ x_2\ b_2\ f_2 
    \ |\ a_2\ \spec{d}\ \myunderline{y}\ \spec{e}\ \myunderline{x\ b}\ \spec{f}\ z_2\ c_2\ | \\
\downarrow &\mynote{(d,e,f)}{Internal triple of $\B2$}\\
I_{6}&= |\ a'\ e_2\ a\ d_2\ y_2\ c'\ \spec{z}\ \myunderline{b'\ c\ x_2\ b_2\ f_2 
    \ |\ a_2}\ \spec{x}\ \myunderline{b}\ \spec{y}\ z_2\ c_2\ | \\
\downarrow &\mynote{(x,y,z)}{\bf Activation of $Y$}\\
I_{5}&= |\ a'\ e_2\ \spec{a}\ \myunderline{d_2\ y_2\ c'}\ \spec{b}\ \myunderline{b'}\ \spec{c}\ x_2\ b_2\ f_2 
    \ |\ a_2\ z_2\ c_2\ | \\
\downarrow &\mynote{(a,b,c)}{Internal triple of $\B1$}\\
I_{4}&= |\ \spec{a'}\ \myunderline{e_2}\ \spec{b'}\ \myunderline{d_2\ y_2}\ \spec{c'}\ x_2\ b_2\ f_2 
    \ |\ a_2\ z_2\ c_2\ | \\
\downarrow &\mynote{(a',b',c')}{Internal triple of $\B1$}\\
I_{3}&= |\ \spec{d_2}\ \myunderline{y_2}\ \spec{e_2}\ \myunderline{x_2\ b_2}\ \spec{f_2} 
    \ |\ a_2\ z_2\ c_2\ | \\
\downarrow &\mynote{(d_2,e_2,f_2)}{Internal triple of $\B1$}\\
I_{2}&= |\ \spec{x_2}\ \myunderline{b_2}\ \spec{y_2} 
    \myunderline{\ |\ a_2}\ \spec{z_2}\ c_2\ | \\
\downarrow &\mynote{(x_2,y_2,z_2)}{\bf Activation of $X_2$}\\
I_{1}&= |\ \emptyinst  
    \ |\ \spec{a_2\ b_2\ c_2}\ | \\
\downarrow &\mynote{(a_2,b_2,c_2)}{Internal triple of $\B2$}\\
I_{0}&= |\ \emptyinst  
    \ |\ \emptyinst\ | = \emptyinst
\end{align*}
\caption{%
\label{fig:exToyAssembling} 3DT-collapsibility of the assembling of basic blocks
${[X_1,X_2]=\blockFont{var}(Y)}$ and ${Y=\blockFont{or}(X_1,X_2)}$.
For each 3DT-step, the three elements that are deleted from the alphabet are in bold, 
the elements that are swapped by the corresponding transposition are underlined. Vertical bars
give the limits of the blocks in the 2-block-decomposition, and dot symbols are omitted.%
}
\end{figure}

\begin{lem}\label{lem:completeBD}%
Let $I'$ be a 3DT-instance with an $l$-block-decomposition $\BB'$, such that $(I',\BB')$ is obtained from 
an assembling of basic blocks $(I,\BB)$ after any number of 3DT-steps, i.e. 
there exist $k\geq 0$ triples $(d_i,e_i,f_i)$, $i\in\intvl1k$, 
such that $(I,\BB)\DT{d_1,e_1,f_1}\cdots\DT{d_k,e_k,f_k}(I',\BB')$.

Then $(I',\BB')$ is a valid context.
Moreover, if the set of variables of $(I',\BB')$ is empty, then $I'$ is 3DT-collapsible.
\end{lem}

\begin{proof}
Write $\mathcal A$ the set of variables used to define $(I,\BB)$. We write $I=\inst$ and $I'=\instp$.
We prove that $(I',\BB')$ is a valid context by induction on $k$ 
(the number of 3DT-steps between $(I,\BB)$ and $(I',\BB')$). 
We also prove that for each $h\in\intvl1l$, $\Bp h$ appears as a node in
the behavior graph of $\B h$.

Suppose first that $k=0$. We show that the set of external triples 
of $(I,\BB)=(I',\BB')$ can be partitioned into 
valid variables, namely into $\mathcal A$. Indeed, from the 
definition of each block, for each $\sigma\in\Sigma$, 
$\sigma$ is either part of an internal triple, or appears in a variable 
$A\in\mathcal A$. Conversely, for each $A=[(a,b,c),(x,y,z)]
\in \mathcal A$, $b$, $x$ and $y$ appear in the block having $A$ for 
output, and $a$, $c$ and $z$ appear in the block having $A$ for input.
Hence $(a,b,c)$ and $(x,y,z)$ are indeed two external triples of $(I,\BB)$. 
Hence each variable satisfies conditions \Cii\ and \Ciii\ of Definition~\ref{def:validVar}.
Conditions \Civ\ and \Cv\ can be checked in the definition of each block: we have, for 
each output variable, $y\pre x$, and for each input variable, 
$a\pre z\pre c$. 
Finally, each $\B h$ appears in its own behavior graph.

Suppose now that $(I',\BB')$ is obtained from $(I,\BB)$ after $k$ 3DT-steps, $k>0$. Then there 
exists a 3DT-instance with an $l$-block-decomposition $(I'',\BB'')$ such that:
\begin{equation*}
(I,\BB)\DT{d_1,e_1,f_1}\cdots\DT{d_{k-1},e_{k-1},f_{k-1}}(I'',\BB'')\DT{d_k,e_k,f_k}(I',\BB').
\end{equation*} 
Consider $h\in\intvl1l$. By induction hypothesis,
since $\BB''_h$ is in a valid context $(I'',\BB'')$,
then, depending on $(d_k,e_k,f_k)$, either $\BB'_h=\BB''_h$, either there is an arc
from $\BB''_h$ to $\BB'_h$ in the behavior graph. Hence $\BB'_h$ is indeed a node
in this graph. By Property~\ref{prop:keepCiiCiii}, we know that the set of external triples of 
$(I',\BB')$ can be partitioned into variables satisfying conditions 
\Cii\ and \Ciii\ of Definition~\ref{def:validVar}.
Hence we need to prove that each variable satisfies conditions \Civ\ and \Cv: 
we verify, for each node of each behavior graph, 
that $x\pre y\Rightarrow x\preCons b\preCons y$ (resp. $a\pre z\pre c$) for 
each output (resp. input) variable $A=[(a,b,c),(x,y,z)]$ of the block. 
This achieves the induction proof.

We finally need to consider the case where the set of variables of $(I',\BB')$ is empty. Then 
for each $h\in\intvl1l$ we either have $\Bp h=\emptyinst$, or $\Bp h=a_h\ b_h\ c_h$ for some internal triple
$(a_h,b_h,c_h)$ (in the case where $\B h$ is a block $\blockFont{or}$). 
Then $(I',\BB')$ is indeed 3DT-collapsible: simply follow in any order the 3DT-step $\DT{a_h,b_h,c_h}$ for 
each remaining triple $(a_h,b_h,c_h)$.
\end{proof}


\clearpage
\subsection{Construction}

\label{sect:assemblage}
Let $\phi$ be a boolean formula, over the boolean variables $x_1,\ldots,x_m$, given in conjunctive normal form:  
$\phi=C_1\wedge C_2 \wedge \ldots \wedge C_\gamma$. 
Each clause $C_c$ ($c\in\intvl1\gamma$) is the disjunction of a number of literals, $x_i$ or $\neg x_i$, $i\in\intvl1m$.
We write $q_i$ (resp.~$\bar q_i$) the number of occurrences of the literal $x_i$ (resp.~$\neg x_i$) in $\phi$, $i\in\intvl1m$. 
We also write $k(C_c)$ the number of literals appearing in the clause $C_c$, $c\in\intvl1\gamma$.
We can assume that $\gamma\geq 2$, that for each $c\in\intvl1\gamma$, we have $k(C_c)\geq 2$, 
and that for each $i\in\intvl1m$, $q_i\geq 2$ and $\bar q_i\geq 2$. 
(Otherwise, we can always add clauses of the form 
$(x_i\vee \neg x_i)$ to $\phi$, 
or duplicate the literals appearing in the clauses $C_c$ such that $k(C_c)=1$.)
In order to distinguish variables of an $l$-block-decomposition from $x_1,\ldots,x_m$, 
we always use the term \emph{boolean variable} for the latter.

The 3DT-instance $I_\phi$ is defined as an assembling of basic blocks: 
we first define a set of variables, then we list the blocks of which
the word representation of $I_\phi$ is the concatenation.
It is necessary that each variable is part of the input (resp. the output) 
of exactly one block. 
Note that the relative order 
of the blocks is of  no importance. We simply try, for readability 
reasons, to ensure that the source of a variable appears before its 
target, whenever possible. We say that a variable \emph{represents} a term, 
i.e. a literal, clause or formula,
if it can be activated only if this term is true 
(for some fixed assignment of the boolean variables), 
or if $\phi$ is satisfied by this assignment. We also say that a block 
\emph{defines} a variable if it is its source block.

The construction of $I_\phi$ is done as follows (see Figure~\ref{fig:schemaGlobal} for an example):
\begin{itemize}
\item Create a set of variables:
\begin{itemize}
\item For each $i\in\intvl1m$, create $q_i+1$ variables representing $x_i$: $X_i$ and $X_i^j$, $j\in\intvl{1}{q_i}$, and  
$\bar q_i+1$ variables representing $\neg x_i$: $\bar X_i$ and $\bar X_i^j$, $j\in\intvl{1}{\bar q_i}$.
\item For each $c\in\intvl1\gamma$, create a variable $\Gamma_c$ representing the clause $C_c$.
\item Create $m+1$ variables, $A_\phi$ and $A_\phi^i$, $i\in\intvl1m$, representing the formula $\phi$.
We will show that $A_\phi$ has a key role in the construction: 
it can be activated only if $\phi$ is satisfiable, and, once activated, 
it allows every remaining variable to be activated.
\item We also use a number of intermediate variables, with names $U$, $\bar U$, $V$, $W$ and $Y$. 
\end{itemize}

\item Start with an empty 3DT-instance $\emptyinst$, and add blocks successively:
\begin{itemize}
\item For each $i\in\intvl1m$, add the following $q_i+\bar q_i-1$ blocks defining the variables $X_i$, $X_i^j$ ($j\in\intvl1{q_i}$), 
and $\bar X_i$, $\bar X_i^j$ ($j\in\intvl1{\bar q_i}$):
\begin{align*}
   &&&\makebox[1pt][c]{$[X_i,\bar X_i]= \blockFont{var}(A_\phi^i)$} \\
	[X_i^1,U_i^2] &= \blockFont{copy}(X_i)&&&
	[\bar X_i^1,\bar U_i^2] &= \blockFont{copy}(\bar X_i)\\
	[X_i^2,U_i^3] &= \blockFont{copy}(U_i^2)&&&
	[\bar X_i^2,\bar U_i^3] &= \blockFont{copy}(\bar U_i^2) \\
	&\:\ \vdots&&&&\:\ \vdots\tag{${*}$}\label{eq:group1}\\
	[X_i^{q_i-2},U_i^{q_i-1}] &= \blockFont{copy}(U_i^{q_i-2})&&&
&\:\ \vdots\\
	[X_i^{q_i-1},X_i^{q_i}] &= \blockFont{copy}(U_i^{q_i-1})&&&
	[\bar X_i^{\bar q_i-2},\bar U_i^{\bar q_i-1}] &= \blockFont{copy}(\bar U_i^{\bar q_i-2})\\
&&&&
	[\bar X_i^{\bar q_i-1},\bar X_i^{\bar q_i}] &= \blockFont{copy}(\bar U_i^{\bar q_i-1})
\end{align*}
\item For each $c\in\intvl1\gamma$, 
let $C_c=\lambda_1\vee \lambda_2\vee \ldots \vee \lambda_k$, with $k=k(C_c)$. 
Let each $\lambda_p$, $p\in\intvl1k$, be the $j$-th occurrence of a literal $x_i$ or $\neg x_i$, 
for some $i\in\intvl1m$ and $j\in\intvl1{q_i}$ (resp. $j\in\intvl1{\bar q_i}$). 
We respectively write $L_p=X_i^j$ or $L_p=\bar X_i^j$. 
We add the following $k-1$ blocks defining $\Gamma_c$:
\begin{align*}
	V_c^2&= \blockFont{or}(L_1,L_2)\\
	V_c^3&= \blockFont{or}(V_c^2,L_3)\\
	& \ \:\vdots \tag{${*}{*}$}\label{eq:group2}\\
	V_c^{k-1}&= \blockFont{or}(V_c^{k-2},L_{k-1})\\
	\Gamma_c&= \blockFont{or}(V_c^{k-1},L_k)
\end{align*}
\item Since $\phi=C_1\wedge C_2\wedge \ldots \wedge C_\gamma$, the formula variable $A_\phi$ is defined by the following $\gamma-1$ blocks:
\begin{align*}
	W_2&= \blockFont{and}(\Gamma_1,\Gamma_2)\\
	W_3&= \blockFont{and}(W_2,\Gamma_3)\\
	& \ \:\vdots \tag{${*}{*}{*}$}\label{eq:group3}\\
	W_{\gamma-1}&= \blockFont{and}(W_{\gamma-2},\Gamma_{\gamma-1})\\
	A_\phi&= \blockFont{and}(W_{\gamma-1},\Gamma_l)
\end{align*}
\item The $m$ copies $A^1_\phi,\ldots,A^m_\phi$ of $A_\phi$ are defined with the following $m-1$ blocks:
\begin{align*}
	[A_\phi^1,Y_2] &= \blockFont{copy}(A_\phi)\\
	[A_\phi^2,Y_3] &= \blockFont{copy}(Y_2)\\
	& \ \:\vdots  \tag{${*}{*}{*}{*}$}\label{eq:group4}\\
	[A_\phi^{m-2},Y_{m-1}] &= \blockFont{copy}(Y_{m-2})\\
	[A_\phi^{m-1},A_\phi^{m}] &= \blockFont{copy}(Y_{m-1})\\
\end{align*}
\end{itemize}
\end{itemize}

\begin{figure}
\begin{center}
\makebox[1pt][c]{
\includegraphics{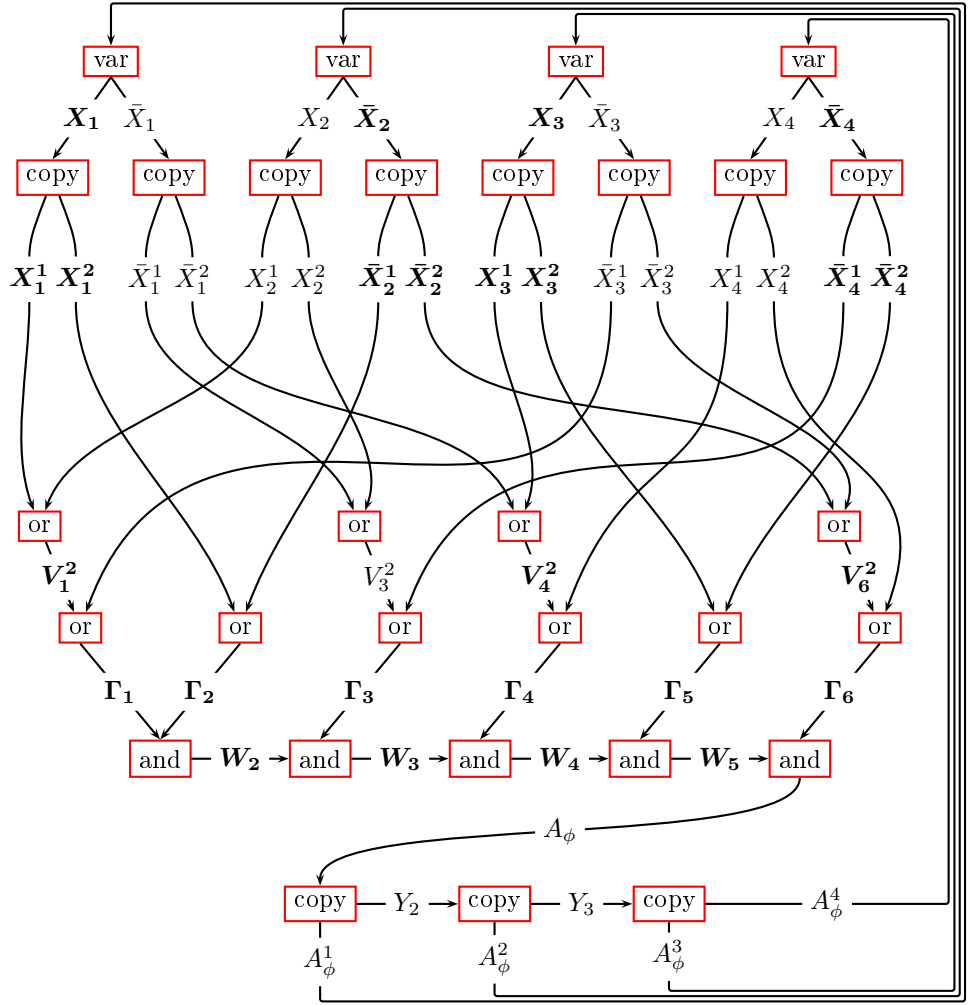}
}
\end{center}
\caption{\label{fig:schemaGlobal}
Schematic diagram of the blocks defining $I_\phi$ for 
$\phi=
{(x_1\vee x_2\vee \neg x_3)}\wedge
{(x_1\vee \neg x_2)}\wedge
{(\neg x_1\vee x_2\vee \neg x_4)}\wedge
{(\neg x_1\vee x_3\vee x_4)}\wedge
{(x_3\vee \neg x_4)}\wedge
{(\neg x_2\vee \neg x_3\vee x_4)}$. 
For each variable, we draw an arc between its source and target block.
Note that $\phi$ is satisfiable (e.g. with the assignment 
$x_1=x_3=true$ and $x_2=x_4=false$).
A set of variables that can be activated before $A_\phi$ is in bold, 
they correspond to the terms being true in $\phi$ for the assignment
$x_1=x_3=true$ and $x_2=x_4=false$.
}
\end{figure}
%


\subsection{Main Result}

\begin{thm} \label{thm:satisfiability}
Let $\phi$ be a boolean formula, and $I_\phi$ the 3DT-instance defined 
in Section~\ref{sect:assemblage}. The construction of $I_\phi$ is 
polynomial in the size of $\phi$,  
and $\phi$ is satisfiable iff $I_\phi$ is 3DT-collapsible.
\end{thm}
\begin{proof}
The polynomial time complexity of the construction of $I_\phi$ is trivial. 
We use the same notations as in the construction, with $\BB$ the block decomposition of $I_\phi$. 
One can easily check, in~\eqref{eq:group1},~\eqref{eq:group2}, 
\eqref{eq:group3} and~\eqref{eq:group4}, each variable has exactly one source block
and one target block. Then, by Lemma~\ref{lem:completeBD},
we know that $(I_\phi,\BB)$ is a valid context, and remains so after any number of 3DT-steps, 
hence properties~\ref{prop:block_copy2}, \ref{prop:block_and}, \ref{prop:block_or} and \ref{prop:block_var} 
are satisfied by respectively 
each block $\blockFont{copy}$, $\blockFont{and}$, $\blockFont{or}$ and $\blockFont{var}$ of $I_\phi$.

\framebox{$\Rightarrow$}
Assume first that $\phi$ is satisfiable. Consider a truth assignment 
satisfying $\phi$: let $P$ be the set of indices $i\in\intvl1m$
such that $x_i$ is 
assigned to true. 
Starting from $I_\phi$, we can follow a path of 
3DT-steps that activates all the variables of $I_\phi$ in 
the following order:
\begin{itemize}
\item For $i\in\intvl1m$, if $i\in P$, activate $X_i$ in the corresponding 
block $\blockFont{var}$ in~\eqref{eq:group1}. Then, with the blocks $\blockFont{copy}$, 
activate successively all intermediate variables $U_i^j$ 
for $j=2$ to $q_i-1$, and variables $X^j_i$ for $j\in\intvl1{q_i}$.

Otherwise, if $i\notin P$, activate $\bar X_i$,  
all intermediate variables $\bar U_i^j$ for $j=2$ to $\bar q_i-1$, 
and the variables $\bar X^j_i$ for $j\in\intvl1{\bar q_i}$

\item For each $c\in\intvl1\gamma$, let $C_c=\lambda_1\vee \lambda_2\vee \ldots \vee \lambda_k$, with $k=k(C_c)$. 
Since $C_c$ is true with the selected truth assignment, at least one literal 
$\lambda_{p_0}$, ${p_0}\in\intvl1k$, is true.
If $\lambda_{p_0}$ is the $j$-th occurrence of a literal 
$x_i$ or $\neg x_i$, then the corresponding variable $L_{p_0}$ 
($L_{p_0}=X_i^j$ or $L_{p_0}=\bar X_i^j$) has been activated previously.
Using the blocks $\blockFont{or}$ in~\eqref{eq:group2},
 we activate successively each intermediate variable $V_c^{p}$ for 
$p={p_0}$ to $p=k-1$, and finally we activate the variable~$\Gamma_c$.

\item Since all variables $\Gamma_c$, $c\in\intvl1\gamma$, have been activated, 
using the blocks $\blockFont{and}$ in~\eqref{eq:group3},
we activate each intermediate variable $W_c$ for $c=2$ to $c=\gamma-1$, 
and the formula variable $A_\phi$.

\item With the blocks $\blockFont{copy}$  in~\eqref{eq:group4}, 
we activate successively all the intermediate variables $Y_i$, $i\in\intvl2{m-1}$ and the $m$ copies 
$A_\phi^1,\ldots,A_\phi^m$ of $A_\phi$.

\item For $i\in\intvl1m$, since the variable $A^i_\phi$ has been activated, we 
activate in the block $\blockFont{var}$ of~\eqref{eq:group1} 
the remaining variable $X_i$ or $\bar X_i$. We also activate 
all its copies and corresponding intermediate variables $U_i^j$ or $\bar U_i^j$.

\item For $c\in\intvl1\gamma$, in~\eqref{eq:group2}, since all variables $L_p$ 
have been activated, we
activate the remaining intermediate variables $V_c^{p}$.

\item At this point every variable has been activated. Using 
again Lemma~\ref{lem:completeBD}, we know that the resulting instance is 3DT-collapsible, 
and can be reduced down to the empty 3DT-instance $\emptyinst$.
\end{itemize}
Hence $I_\phi$ is 3DT-collapsible.

\framebox{$\Leftarrow$}
Assume now that $I_\phi$ is 3DT-collapsible: we consider a 
sequence of 3DT-steps reducing $I_\phi$ to $\emptyinst$.
This sequence gives a total order on the set of variables: the order
in which they are activated.
We write 
$Q$ the set of variables activated before $A_\phi$, and $P\subseteq\intvl1m$ the set of 
indices $i$ such that $X_i\in Q$ (see the variables in bold in Figure~\ref{fig:schemaGlobal}).
We show that the truth assignment defined by 
$(x_i=true \Leftrightarrow i\in P)$ satisfies the formula~$\phi$. 
\begin{itemize}

\item For each $i\in\intvl1m$, $A_\phi^i$ cannot belong to $Q$, using 
the property of the block $\blockFont{copy}$ in~\eqref{eq:group4} 
(each $A_\phi^i$ can only be activated after $A_\phi$). 
Hence, with the block 
$\blockFont{var}$ in~\eqref{eq:group1},  
we have ${\bar X_i\in Q \Rightarrow} X_i\notin Q$. 
Moreover, with the block 
$\blockFont{copy}$, we have 
\begin{align*}
\forall 1\leq j\leq q_i,&\quad X_i^j\in Q\Rightarrow X_i\in Q \tag{a}\label{eq:main_a}\\
\forall 1\leq j\leq \bar q_i,&\quad \bar X_i^j\in Q \Rightarrow \bar X_i\in Q \Rightarrow X_i\notin Q \tag{b}\label{eq:main_b}
\end{align*}

\item Since $A_\phi$ is defined in a block $A_\phi=\blockFont{and}(W_{\gamma-1},\Gamma_\gamma)$ in~\eqref{eq:group3}, 
we necessarily have $W_{\gamma-1}\in Q$ and $\Gamma_\gamma\in Q$. 
Likewise, since $W_{\gamma-1}$ is defined by $W_{\gamma-1}=\blockFont{and}(W_{\gamma-2},\Gamma_{\gamma-1})$, we also have
 $W_{\gamma-2}\in Q$ and $\Gamma_{\gamma-1}\in Q$. 
Applying this reasoning recursively, 
we have $\Gamma_c\in Q$ for each $c\in\intvl1\gamma$.

\item For each $c\in\intvl1\gamma$, consider the clause 
$C_c=\lambda_1\vee \lambda_2\vee \ldots \vee \lambda_k$, with $k=k(C_c)$. 
Using the property of the block $\blockFont{or}$ in~\eqref{eq:group2}, 
there exists some $p_0\in\intvl1k$ such that the 
variable $L_{p_0}$ is activated before $\Gamma_c$: hence $L_{p_0}\in Q$. 
If the corresponding literal $\lambda_{p_0}$ is the $j$-th occurrence of $x_i$ 
(respectively, $\neg x_i$), then $L_{p_0}=X_i^j$ (resp., $L_{p_0}=\bar X_i^j$), 
thus by~\eqref{eq:main_a} (resp. \eqref{eq:main_b}), $X_i\in Q$ (resp., $X_i\notin Q$), 
and consequently $i\in P$ (resp., $i\notin P$). 
In both cases, the literal $\lambda_{p_0}$ is true in the 
truth assignment defined by $(x_i=true \Leftrightarrow i\in P)$. 
\end{itemize}
If $I_\phi$ is 3DT-collapsible, we have found  a truth assignment such that at 
least one literal is true in each clause of the formula $\phi$, and thus $\phi$ is 
satisfiable.
\end{proof}


\newpage
\section{Sorting by Transpositions is NP-Hard}
As noted previously, there is no guarantee that any 3DT-instance 
$I$ has an equivalent permutation
$\pi$. However, with the following theorem, we show that such a 
permutation can be found in the special case
of assemblings of basic blocks, which is the case 
we are interested in, in order to complete our reduction.

\newcommand{\thisI}{I}
\newcommand{\thisPi}{\pi_I}

\begin{thm} \label{thm:createEquiv}
Let $I$ be a 3DT-instance of span $n$ with $\BB$ an $l$-block-decomposition such that $(I,\BB)$
is an assembling of basic blocks.
Then there exists a permutation $\pi_I$, 
computable in polynomial time in $n$, such that $I\sim \pi_I$.
\end{thm}

An example of the construction of $\pi_I$ for the 3DT-instance defined in Example~\ref{ex:ToyAssembling} 
is given in Figure~\ref{fig:exToyEquiv}.

\begin{proof}
Let $\mathcal A$ be the set of variables of the $l$-block-decomposition $\BB$ of $\thisI=\inst$. 
Let $n$ be the span of $\thisI$, and $L$ its domain. Note that $L=\intvl1n$.
For any $h\in\intvl1l$, we write $\ni(\B h)$ (resp. $\no(\B h)$) the number of 
input (resp. output) variables of $\B h$. We also define two integers $p_h,q_h$ by:
\begin{eqnarray*}
&& p_1=0 \\
\forall h\in\intvl1l, && q_h=p_h+t_h-s_h+3(\ni(\B h)-\no(\B h))\\
\forall h\in\intvl2{l},&& p_{h}=q_{h-1}
\end{eqnarray*}
The permutation $\thisPi$ will be defined such that $p_h$ and $q_h$ have the following property for any $h\in\intvl1l$: 
$\thisPi(s_h)=p_h$, and $\thisPi(t_h)=q_h$.

We also define two applications $\alpha,\beta$ over the set $\mathcal A$ of variables. 
The permutation $\thisPi$ will be defined so that, for any variable $A=[(a,b,c),(x,y,z)]$, 
we have $\thisPi(\psi(a)-1)=\alpha(A)$ and $\thisPi(\psi(z)-1)=\beta(A)$. 
In order to have this property, $\alpha$ and $\beta$ are defined as follows.

For each $h\in \intvl1l$: \begin{itemize}
\item
       If $\B h$ is a block of the kind $[A_1,A_2]=\blockFont{copy}(A)$, 
 define
\begin{equation*}
\alpha(A)=p_h,\ \beta(A)=p_h+4.
\end{equation*} 

\item
       If $\B h$ is a block of the kind $A=\blockFont{and}(A_1,A_2)$, 
 define
\begin{equation*}
\alpha(A_1)=p_h,\ \beta(A_1)=p_h+7,\ \alpha(A_2)=p_h+3,\ \beta(A_2)=p_h+9.\end{equation*} 

 \item
       If $\B h$ is a block of the kind $A=\blockFont{or}(A_1,A_2)$, 
 define
\begin{equation*}
\alpha(A_1)=p_h,\ \beta(A_1)=p_h+13,\ \alpha(A_2)=p_h+3,\ \beta(A_2)=p_h+16.
\end{equation*}

\item
       If $\B h$ is a block of the kind $[A_1,A_2]=\blockFont{var}(A)$, 
 define
\begin{equation*}
\alpha(A)=p_h+5,\  \beta(A)=p_h+9.
\end{equation*}

\end{itemize}

Note that for every $A\in\mathcal A$, $\alpha(A)$ and $\beta(A)$ are defined once and only once,
depending on the kind of the block $\B{target(A)}$.
The permutation $\thisPi$ is designed in such a way that 
the image by $\thisPi$ of an interval $\intvl{s_h+1}{t_h}$ is essentially the interval $\intvl{p_h+1}{q_h}$.
However, there are exceptions: namely, for each variable $A$, the integers $\alpha(A)+1,\alpha(A)+2,\beta(A)+1$,
which are included in $\intvl{p_{target(A)}+1}{q_{target(A)}}$, are in the image of $\intvl{s_{source(A)}+1}{t_{source(A)}}$.
This is formally described as follows.
For each $h\in\intvl1k$ we define a set $P_h$ by:
\begin{eqnarray*}
P_h=\intvl{p_h+1}{q_h}
&\cup& \bigcup_{A\text{ output of }\B h} \{\alpha(A)+1,\alpha(A)+2,\beta(A)+1\} \\
&-& \bigcup_{A\text{ input of }\B h} \{\alpha(A)+1,\alpha(A)+2,\beta(A)+1\} 
\end{eqnarray*}
We note that the sets $\{\alpha(A)+1,\alpha(A)+2,\beta(A)+1\}$ are distinct for different 
variables $A$, and are each included in their respective interval $\intvl{p_{target(A)}+1}{q_{target(A)}}$.
Hence for any $h\in\intvl1l$, we have 
$\lvert P_h\rvert=q_h-p_h + 3\no(\B h)- 3\ni(\B h) = t_h-s_h$. 
Moreover, the sets $P_h$, $h\in\intvl1l$, form a partition of the set $\intvl1n$.

We can now create the permutation $\thisPi$. The image of $0$ is $0$, and for each $h_0$ from $1$ to $l$,
we define the restriction of $\thisPi$ over $\intvl{s_{h_0}+1}{t_{h_0}}$ as a permutation of $P_{h_0}$,
with the constraint that $\thisPi(t_{h_0})=q_{h_0}$.
Note that, if this condition is fulfilled, then 
we can assume $\thisPi(s_{h_0})=p_{h_0}$, 
since, if ${h_0}=1$, $\thisPi(s_1)=\thisPi(0)=0=p_1$, 
and if ${h_0}>1$, $\thisPi(s_{h_0})=\thisPi(t_{h_0-1})=q_{h_0-1}=p_{h_0}$.

The definition of $\thisPi$ over
each kind of block is given in Table~\ref{tab:result}. This table is obtained
by applying the following rules, 
until $\thisPi(u)$ is defined for all $u\in \intvl{s_{h_0}+1}{t_{h_0}}$.

\begin{table}
\caption{Definition of $\thisPi$ over an interval $\intvl{s_{h_0}+1}{t_{h_0}}$, where $\B{h_0}$ is one of the blocks $\blockFont{copy}$, $\blockFont{and}$,
$\blockFont{or}$, $\blockFont{var}$. 
We write $s=s_{h_0}$ and $p=p_{h_0}$. We give the line $\psi^{-1}(u)$ 
as a reminder of the definition of each block. 
We also add a column for $u=s$ as a reminder of the fact that $\thisPi(s)=p$.
\label{tab:result}}
\begin{itemize} \addtolength{\arraycolsep}{-2pt}
\item If $\B{h_0}$ is a block of the kind $[A_1,A_2]=\blockFont{copy}(A)$, 
we write $\alpha_1,\beta_1,\alpha_2,\beta_2$ the respective values of $\alpha(A_1),\beta(A_1),\alpha(A_2),\beta(A_2)$.
$$\begin{array}{c@{=}ccccccccccccccccc}
u     &\ \ s\ \  &s{+}1&s{+}2&s{+}3&s{+}4&s{+}5&s{+}6&s{+}7&s{+}8&s{+}9&s{+}10&s{+}11&s{+}12\\
\thisPi(u)     &p& \alpha_1{+}2 &p{+}8&p{+}4&p{+}3& \alpha_2{+}2 &p{+}7& \beta_1{+}1& \alpha_1{+}1 &p{+}6 & \beta_2{+}1 & \alpha_2{+}1 &p{+}9\\
\vspace{1em}\psi^{-1}(u)     &&a& y_1& e & z& d & y_2& x_1& b_1& c& x_2& b_2& f 
\end{array}$$
\item If $\B{h_0}$ is a block of the kind $A=\blockFont{and}(A_1,A_2)$, 
we write $\alpha,\beta$ the respective values of $\alpha(A),\beta(A)$.
$$\begin{array}{c@{=}ccccccccccccccccc}
u    &\ \ s\ \ &s{+}1&s{+}2&s{+}3&s{+}4&s{+}5&s{+}6&s{+}7&s{+}8&s{+}9 &s{+}10&s{+}11&s{+}12\\
\thisPi(u)    &p&p{+}14 &p{+}7&p{+}3& p{+}13 &p{+}9&p{+}6 &\alpha{+}2& p{+}12& p{+}11 & \beta{+}1 &\alpha{+}1  & p{+}15\\
\vspace{1em}\psi^{-1}(u)     &&a_1& e& z_1& a_2& c_1& z_2& d& y& c_2 & x& b& f
\end{array}$$
\item If $\B{h_0}$ is a block of the kind $A=\blockFont{or}(A_1,A_2)$, 
we write $\alpha,\beta$ the respective values of $\alpha(A),\beta(A)$.
$$\begin{array}{c@{=}cccccccccccccccccccccccc}
u    &s&s{+} 1&s{+}2&s{+}3&s{+}4&s{+}5&s{+}6&s{+}7&s{+}8&s{+}9\\
\thisPi(u)    &p&p{+}7 &p{+}13 &p{+}3 &p{+}9 &\alpha{+}2 &p{+}12 &p{+}11 &\beta{+}1 &\alpha{+}1 \\
\vspace{1em}\psi^{-1}(u)    &&a_1& b'& z_1& a_2& d& y& a'& x& b \\
u    &s{+}10&s{+}11&s{+}12&s{+}13&s{+}14&s{+}15\\
\thisPi(u)     &p{+}16 &p{+}6 &p{+}15 &p{+}10 &p{+}8 &p{+}18\\
\vspace{1em}\psi^{-1}(u)    & f& z_2& c_1& e& c'& c_2
\end{array}$$
\item If $\B{h_0}$ is a block of the kind $[A_1,A_2]=\blockFont{var}(A\zz)$, 
we write $\alpha_1,\beta_1,\alpha_2,\beta_2$ the respective values of $\alpha(A_1),\beta(A_1),\alpha(A_2),\beta(A_2)$.
$$\begin{array}{c@{=}cccccccccccccccccccccccc}
u    &s&s{+} 1&s{+}2&s{+}3&s{+}4&s{+}5&s{+}6&s{+}7&s{+}8&s{+}9\\
\thisPi(u)    &p& \alpha_1{+}2 &p{+}5 &p{+}3 &\alpha_2{+}2 &p{+}12 &p{+}1 &p{+}14 &p{+}4 &\beta_1{+}1 \\
\vspace{1em}\psi^{-1}(u)    &&d_1& y_1& a\zz& d_2& y_2& e_1& a'& e_2& x_1\\
u    &s{+}10&s{+}11&s{+}12&s{+}13&s{+}14&s{+}15&s{+}16&s{+}17&s{+}18\\
\thisPi(u)    &\alpha_1{+}1 &p{+}13 &p{+}9 &p{+}8 &p{+}2 &p{+}11 &\beta_2{+}1&\alpha_2{+}1 &p{+}15 \\
\vspace{1em}\psi^{-1}(u)    & b_1& f_1& c'& z\zz& b'& c\zz& x_2& b_2& f_2
\end{array}$$
\end{itemize}\addtolength{\arraycolsep}{-2pt}
\end{table}

\begin{align*}
\forall A=[(a,b,c),(x,y,z)] &\text{ input variable of }\B{h_0} \\
\thisPi(\psi(z))&=\alpha(A)+3 \tag{$R_1$}\label{eqn:Ri1}\\
\thisPi(\psi(c))&=\beta(A)+2 \tag{$R_2$}\label{eqn:Ri2}\\
\forall A=[(a,b,c),(x,y,z)] &\text{ output variable of }\B{h_0} \\
\thisPi(\psi(x))&=\beta(A)+1 \tag{$R_3$}\label{eqn:Ro1}\\
\thisPi(\psi(b))&=\alpha(A)+1 \tag{$R_4$}\label{eqn:Ro2}\\
\forall u \in \intvl{s_{h_0}+1}{t_{h_0}} \text{ such that }&\succIph^{-1}(u)\in\intvl{s_{h_0}+1}{t_{h_0}}\\
\thisPi(u)&=\thisPi(\succI^{-1}(u)-1)+1 \tag{$R_5$}\label{eqn:Rdef}\\
\end{align*}

We can see in Table~\ref{tab:result} that rules~\eqref{eqn:Ri1} and~\eqref{eqn:Ri2} 
indeed apply to every input variable, and 
rules~\eqref{eqn:Ro1} and~\eqref{eqn:Ro2} apply to every output variable. Moreover:
\begin{equation*}
\begin{array}{c}
\text{Rule~\eqref{eqn:Rdef} applies to every $u\in \intvl{s_{h_0+1}}{t_{h_0}}$ such that}\\
u\notin \{\psi(b),\psi(c),\psi(x),\psi(z) \mid A=[(a,b,c),(x,y,z)]
\text{ input/output of }\B{h_0}\}. 
\end{array}
\tag{$P_1$}\label{eqn:Pu}
\end{equation*}
A simple case by case analysis shows that the following properties are also satisfied.
\begin{equation*}
\thisPi \text{ defines a bijection from $\intvl{s_{h_0}+1}{t_{h_0}}$ to $P_{h_0}$ such that } \thisPi(t_{h_0})=q_{h_0}\tag{$P_2$}\label{eqn:Pg1}
\end{equation*}%
\begin{align*}
\forall A=[(a,b,c),(x,y,z)] &\text{ input variable of }\B{h_0}, \\
\thisPi(\psi(a)-1)&=\alpha(A) \tag{$P_3$}\label{eqn:Pi1}\\
\thisPi(\psi(z)-1)&=\beta(A) \tag{$P_4$}\label{eqn:Pi2}\\
\forall A=[(a,b,c),(x,y,z)] &\text{ output variable of }\B{h_0}, \\
\thisPi(\psi(y)-1)&=\alpha(A)+2 \tag{$P_5$}\label{eqn:Po1}\\
\thisPi(\psi(b)-1)&=\beta(A)+1 \tag{$P_6$}\label{eqn:Po2}
\end{align*}

Now that we have defined the permutation $\thisPi$,
we need to show that $\thisPi$ is equivalent to $\thisI$. 
Following Definition~\ref{def:equiv}, we have $\thisPi(0)=0$.
Then, $L=\intvl{1}{n}$, so let us fix any $u\in\intvl{1}{n}$, and 
verify that $\thisPi(u)=\thisPi(\succI^{-1}(u)-1)+1$. 
Let $h$ be the integer such that $u\in\intvl{s_{h}+1}{t_{h}}$.

First consider the most general case, where there is no variable $A=[(a,b,c),(x,y,z)]$ 
such that $u\in\{\psi(b),\psi(c),\psi(x),\psi(z)\}$. Note that 
this case includes $u=\psi(d)$, where
$d$ is part of any internal triple. Then, by Property~\eqref{eqn:Pu}, 
we know that Rule~\eqref{eqn:Rdef} applies to $u$, hence we directly 
have $\thisPi(u)=\thisPi(\succI^{-1}(u)-1)+1$.

Suppose now that, for some variable $A=[(a,b,c),(x,y,z)]$, 
we have $u\in\{\psi(b),\psi(c),\psi(x),\psi(z)\}$. 
  Then Rules~\eqref{eqn:Ri1} and~\eqref{eqn:Ri2}, 
and Properties~\eqref{eqn:Pi1} and~\eqref{eqn:Pi2} 
apply in the target block of $A$.
  Also, Rules~\eqref{eqn:Ro1} and~\eqref{eqn:Ro2}, 
and Properties~\eqref{eqn:Po1} and~\eqref{eqn:Po2} 
apply in the source block of $A$.
  Combining all these equations together, we have:
\begin{align*}
\thisPi(\psi(b))&=\alpha(A)+1=\thisPi(\psi(a)-1)+1 && \text{by }\eqref{eqn:Ro1} \text{ and } \eqref{eqn:Pi1}  \\
\thisPi(\psi(c))&=\beta(A)+2=\thisPi(\psi(b)-1)+1 && \text{by }\eqref{eqn:Ri2} \text{ and } \eqref{eqn:Po1} \\
\thisPi(\psi(x))&=\beta(A)+1=\thisPi(\psi(z)-1)+1 && \text{by }\eqref{eqn:Ro2} \text{ and } \eqref{eqn:Pi2}  \\
\thisPi(\psi(z))&=\alpha(A)+3=\thisPi(\psi(y)-1)+1 && \text{by }\eqref{eqn:Ri1} \text{ and } \eqref{eqn:Po2} 
\end{align*}

For $u=\psi(b)$ (resp. $\psi(c),\psi(x),\psi(z)$), we have $\succI^{-1}(u)=\psi(a)$ (resp. $\psi(b),\psi(z),\psi(y)$). Hence, in all four cases, we have $\thisPi(u)=\thisPi(\succI^{-1}(u)-1)+1$,  which completes the proof that $\thisPi$ is equivalent to $\thisI$.

\end{proof}

\begin{figure}
\addtolength{\arraycolsep}{-3.62pt}
\small
$$\begin{array}{|rcc|cccccccccccccccccc|ccccccccccccccc|}
\hline\multicolumn{3}{|r| }{}&
\multicolumn{18}{c| }{[X_1,X_2]=\blockFont{var}(Y)}&
\multicolumn{15}{c| }{Y=\blockFont{or}(X_1,X_2)}
\\
&&&&&&&&&&&&&&&&&&&&&&&&&&&&&&&&&&&\\
\multicolumn{3}{|c|}{}&
\multicolumn{13}{r}{\text{Input variable (target of): }}&
\multicolumn{5}{l| }{Y}&
\multicolumn{11}{r}{\text{Input variables (target of): }}&
\multicolumn{4}{l| }{X_1,X_2}\\
\multicolumn{3}{|c|}{}&
\multicolumn{13}{r}{\text{Output variables (source of): }}&
\multicolumn{ 5}{l| }{X_1,X_2}&
\multicolumn{11}{r  }{\text{Output variable (source of): }}&
\multicolumn{ 4}{l| }{Y}\\
\multicolumn{3}{|c|}{}&
\multicolumn{18}{c| }{}&
\multicolumn{15}{c| }{}\\
\multicolumn{3}{|c|}{}&
\multicolumn{18}{c|}{
\begin{array}{rlcrl}
s_1&=0  &\qquad& t_1&= 18\\
p_1&=0  && q_1&=p_1+18-3=15
\end{array}
}&
\multicolumn{15}{c|}{
\begin{array}{rlcrl}
s_2&=18 &\qquad& t_2&= 33\\
p_2&=15 &&q_2&=p_2+15+3=33
\end{array}
}
\\ 
&&&&&&&&&&&&&&&&&&&&&&&&&&&&&&&&&&&\\
\multicolumn{3}{|c|}{}&
\multicolumn{18}{c|}{
\begin{array}{rl}
\alpha(Y)&=p_1+5=5 \\
 \beta(Y)&=p_1+9 =9
\end{array}
}&
\multicolumn{15}{c|}{
\begin{array}{rl}
\alpha(X_1)&=p_2=15 \\
 \beta(X_1)&= p_2+13=28\\
\alpha(X_2)&=p_2+3=18\\
\beta(X_2)&=p_2+16=31
\end{array}
}
\\ 
&&&&&&&&&&&&&&&&&&&&&&&&&&&&&&&&&&&\\\hline
\multicolumn{3}{|c|}{}&
\multicolumn{18}{c| }{\text{Definition of $\thisPi$ over $\intvl{s_1+1}{t_1}$}:}&
\multicolumn{15}{c| }{\text{Definition of $\thisPi$ over $\intvl{s_2+1}{t_2}$}:}\\

u&=&0&1&2&3&4&5&6&7&8&9&10&11&12&13&14&15&16&17&18&19&20&21&22&23&24&25&26&27&28&29&30&31&32&33\\
\thisPi(u)&=&0&17&5&3&20&12&1&14&4&29&16&13&9&8&2&11&32&19&15&22&28&18&24&7&27&26&10&6&31&21&30&25&23&33\\
\psi^{-1}(u)&=&&d_1&y_1&a&d_2& y_2& e_1& a'& e_2& x_1& b_1& f_1& c'& z& b'& c& x_2& b_2& f_2&a_1& b''& z_1& a_2& d& y& a''& x& b& f& z_2& c_1& e& c''& c_2\\\hline
\end{array}$$
\caption{
\label{fig:exToyEquiv} Creation of a permutation $\pi_I$ equivalent to the assembling of basic blocks
$I=\inst$ of span 33 defined in Example~\ref{ex:ToyAssembling}, following the proof of Theorem~\ref{thm:createEquiv}. 
}
\end{figure}

With the previous theorem, we now have all the necessary ingredients 
to prove the main result of 
this paper. 

\begin{thm}
The \textsc{Sorting by Transpositions} problem is \textsf{NP}-hard.
\end{thm}
\begin{proof}
The reduction from SAT is as follows: given any instance $\phi$ of SAT, 
create a 3DT-instance $I_\phi$, being an assembling of basic blocks,
 which is 3DT-collapsible iff $\phi$ is 
satisfiable (Theorem~\ref{thm:satisfiability}).
Then create a 3-permutation $\pi_{I_\phi}$ equivalent to $I_\phi$ 
(Theorem~\ref{thm:createEquiv}). 
The above two steps can be done in polynomial time.
Finally, set $k=d_b(\pi_{I_\phi})/3=n/3$. We then have:
\begin{eqnarray*}
\phi \text{ is satisfiable}&\Leftrightarrow& I_\phi \text{ is 3DT-collapsible}\\
&\Leftrightarrow& d_t(\pi_{I_\phi})=k \text{ (by Theorem~\ref{thm:useEquiv}, since $\pi_{I_\phi}\sim I_\phi$)}\\
&\Leftrightarrow& d_t(\pi_{I_\phi})\leq k \text{ (by Property~\ref{prop:td-bd2})}. 
\end{eqnarray*}
\end{proof}

Note that the permutation $\pi_I$ defined by Theorem~\ref{thm:createEquiv} 
is in fact a 
3-permutation, i.e. a permutation whose cycle graph contains only 3-cycles~\cite{BafnaPevzner95} 
(which is equivalent to saying that the application $succ$ defined by 
$succ(u)=\pi_I^{-1}(\pi_I(u-1)+1)$ has no fixed point, and is such that 
$succ\circ succ \circ succ$ is the identity). Moreover, the number of breakpoints of $\pi_I$ is $d_b(\pi_I)=n$. Hence we have the following corollary.

\begin{corollary}
The following two decision problems \cite{Christie98} are \textsf{NP}-hard:
\begin{itemize}
\item   Given a permutation $\pi$ of $\intvl0n$, is the equality $d_t(\pi)=d_b(\pi)/3$ satisfied? 
\item Given a 3-permutation $\pi$ of $\intvl0n$, is the equality $d_t(\pi)=n/3$ satisfied?
\end{itemize}
\end{corollary}

\section*{Conclusion}

In this paper we have proved that the \textsc{Sorting by Transpositions} 
problem is NP-hard, thus answering a long-standing question. 
 However, a number of questions remain open. 
For instance, does this problem admit a polynomial time approximation scheme? 
We note that the reduction we have provided does not answer this question,
since it is not a linear reduction. 
Indeed, by our reduction, if a formula $\phi$ 
is not satisfiable, it can be seen that we have 
$d_t(\pi_{I_\phi})=d_b(\pi_{I_\phi})/3+1$.

Also, does there exist some relevant parameters for which the problem is fixed parameter tractable? 
A parameter that comes to mind when dealing with the 
transposition distance is the size of the factors
exchanged (e.g., the value $\max\{j-i,k-j\}$ 
for a transposition $\tau_{i,j,k}$). 
Does the problem become tractable if we bound this parameter? 
In fact, the answer to this question is no if we bound only the size of 
the smallest factor, $\min\{j-i,k-j\}$: in our reduction, 
this parameter is upper bounded by 
$6$ for every transposition needed to 
sort $\pi_{I_\phi}$, independently of the formula $\phi$.

\clearpage

\bibliographystyle{plain}
\bibliography{biblio}

\end{document}